\documentclass[11pt,letterpaper]{article}
\usepackage[T1]{fontenc}

\title{Complexity of Strong Popularity in Hedonic Games}

\date{}

\author{
Matan Gilboa\\
\\
University of Oxford
}

\usepackage{amsthm} 
\usepackage{amsmath}
\usepackage{amssymb}
\usepackage{graphicx}
\usepackage{dsfont}
\usepackage{appendix}
\usepackage{enumitem}

\usepackage{algorithm}
\usepackage{algpseudocode}
\usepackage[margin=1in,paper=letterpaper]{geometry}
\usepackage{multirow}
\usepackage{natbib}
\usepackage{setspace}
\usepackage{tikz}
\usetikzlibrary{positioning, shapes, fit, backgrounds, calc, arrows.meta}

\usepackage[unicode, hidelinks]{hyperref}
\usepackage[nameinlink]{cleveref}

\usepackage[most]{tcolorbox}

\newtcolorbox[auto counter, number within=section]{wbox}[1][]{%
  colback=white,
  colframe=black,
  fonttitle=\bfseries,
  title=Problem~\thetcbcounter: #1,
  boxrule=0.5pt,
  boxsep=1pt,        
  left=2pt,          
  right=2pt,
  top=2pt,
  bottom=2pt,
  before skip=4pt,   
  after skip=4pt,    
  arc=0pt            
}

\newtheorem{theorem}{Theorem}[section]
\newtheorem{lemma}{Lemma}

\newtheorem{prop}{Proposition}
\newtheorem{observation}{Observation}
\newtheorem{definition}{Definition}

\crefname{wbox}{Problem}{Problems}
\Crefname{wbox}{Problem}{Problems}
\crefname{prop}{proposition}{propositions}   
\Crefname{prop}{Proposition}{Propositions}  
\crefname{lemma}{lemma}{lemmas}   
\Crefname{lemma}{Lemma}{Lemmas}
\crefname{observation}{observation}{observations}   
\Crefname{observation}{Observation}{Observations}
\crefname{definition}{definition}{definitions}   
\Crefname{definition}{Definition}{Definitions}

\usepackage{xcolor}
\definecolor{mygray}{RGB}{218,215,203}

\usepackage{todonotes}
\presetkeys{todonotes}{inline}{}

\newcommand{\RR}{\ensuremath{\mathbb{R}}}
\DeclareMathOperator{\sgn}{sgn} 

\newcommand{\fCC}{\ensuremath{\mathfrak{C}}} 
\newcommand{\CC}{\ensuremath{\mathcal{C}}} 
\newcommand{\GG}{\ensuremath{\mathcal{G}}} 
\newcommand{\XX}{\ensuremath{\mathtt{x}}} 
\newcommand{\YY}{\ensuremath{\mathtt{y}}} 
 
\newcommand{\OO}{\ensuremath{\mathcal{O}}} 

\newcommand{\bs}{\setminus}
\newcommand{\reduce}{\ensuremath{\leq_P}} 
\newcommand{\bits}[1]{\ensuremath{\{0,1\}^{#1}}}

\newcommand{\complexityclass}[1]{\ensuremath{\bf{#1}}}
\newcommand{\NP}{\complexityclass{NP}}
\newcommand{\coNP}{\complexityclass{coNP}}
\newcommand{\Pclass}{\complexityclass{P}}

\newcommand{\PNP}{\complexityclass{\Pclass^{NP}}}
\newcommand{\PCW}{\complexityclass{PCW}}
\newcommand{\PTW}{\complexityclass{PTW}}
\newcommand{\PMA}{\complexityclass{PMA}}

\newcommand{\SymP}{\complexityclass{S_2^P}}

\newcommand{\ZPPNP}{\complexityclass{ZPP^{NP}}}

\newcommand{\Piclass}[1]{\complexityclass{\Pi_{#1}^P}}

\newcommand{\Sig}[1]{\complexityclass{\Sigma_{#1}^P}}
\newcommand{\SigPi}[1]{\complexityclass{\Sig{#1}\cap\Piclass{#1}}}

\newcommand{\CKTCondorcet}{\textsc{Ckt-Condorcet}}

\newcommand{\ASHG}{\textsc{Ashg-Strong-Popularity}}
\newcommand{\FHG}{\textsc{Fhg-Strong-Popularity}}

\newcommand{\mv}[1]{\ensuremath{m_{\mathbf{v}}(#1)}} 
\newcommand{\countv}[3]{\ensuremath{count_{#3}(#1,#2)}} 
\newcommand{\xx}{\ensuremath{x}} 
\newcommand{\gt}{\ensuremath{g^1}}
\newcommand{\gf}{\ensuremath{g^0}}
\newcommand{\asst}{\ensuremath{a^1}}
\newcommand{\assf}{\ensuremath{a^0}}
\newcommand{\ot}{\ensuremath{o^1}}
\newcommand{\of}{\ensuremath{o^0}}

\newcommand{\ls}{\ensuremath{\ell^a}}

\renewcommand{\lg}{\ensuremath{\ell^g}}
\newcommand{\La}{\ensuremath{L^a}}
\newcommand{\Lg}{\ensuremath{L^g}}

\newcommand{\Fa}{\ensuremath{F^a}}
\newcommand{\Ta}{\ensuremath{T^a}}
\newcommand{\Fg}{\ensuremath{F^g}}
\newcommand{\Tg}{\ensuremath{T^g}}

\newcommand{\wand}{\ensuremath{w}}
\newcommand{\zand}{\ensuremath{z}}
\newcommand{\yand}{\ensuremath{y}}

\newcommand{\Wand}{\ensuremath{W}}
\newcommand{\Zand}{\ensuremath{Z}}
\newcommand{\Yand}{\ensuremath{Y}}

\newcommand{\lw}{\ensuremath{\ell^w}}
\newcommand{\lz}{\ensuremath{\ell^z}}
\newcommand{\ly}{\ensuremath{\ell^y}}

\newcommand{\Lw}{\ensuremath{L^w}}
\newcommand{\Lz}{\ensuremath{L^z}}
\newcommand{\Ly}{\ensuremath{L^y}}

\newcommand{\DD}{\ensuremath{D}}
\newcommand{\GDD}{\ensuremath{GD}}
\newcommand{\ADD}{\ensuremath{AD}}

\newcommand{\MC}{\ensuremath{MC}}
\newcommand{\ga}{\ensuremath{p}}

\begin{document}

\maketitle

\begin{abstract}
In a hedonic game, agents are partitioned into coalitions and have preferences over the coalitions they may join.
A partition is called strongly popular if it strictly beats every other partition in a majority vote among the agents.

We study the complexity of determining the existence of strongly popular partitions in two fundamental classes of hedonic games: additively separable and fractional hedonic games (ASHGs and FHGs), where agents' utilities are determined, respectively, by the sum and average of their valuations for the members of their coalitions. 

We prove that this problem is \PCW{}-complete for ASHGs and for FHGs even with non-negative valuations, where \PCW{} is a subclass of \Sig{2} associated with unambiguous problems (which have at most one witness), recently introduced by \cite{GGKN2025unambiguity}.
These are the first natural complete problems for \PCW{}, in the sense that their formulations do not involve arbitrary Boolean circuits.
Our results resolve open problems posed by \cite{BrBu20a,BG24popularity}.
The two reductions employ substantially different techniques, reflecting both the distinct utility models and the stringent non-negativity constraint in FHGs.

Furthermore, in contrast with FHGs, we show that in ASHGs restricting valuations to be non-negative makes the problem polynomial-time solvable.
\end{abstract}

\newpage

\section{Introduction}
\label{sec:introduction}
We study scenarios in which a set of agents needs to be partitioned into \emph{coalitions}, where each agent expresses a preference order over the coalitions she may be a part of. 
Such settings are motivated by various real-life scenarios such as forming working groups or political parties, establishing research collaborations, international agreements, and customs unions \cite{Ray07a}.
Related applications include community detection and clustering, which are central to the social sciences and machine learning \cite{Newm04a,CLMP22a,csaji2026clustering}.
These are typically studied under the framework of hedonic games---coalition formation games in which agents care only about the coalition they are in, studied, for example, in \cite{DrGr80a,BKS01a,CeRo01a}.

Since the number of coalitions that can be formed from $n$ agents is exponential in $n$, one needs to address the question of how to concisely represent the agents' preferences.
A common approach is to focus on \emph{cardinal} hedonic games, in which the agents assign an individual value to each other agent, and the utility they obtain in a coalition $S$ is some polynomial-time computable aggregation of the values they assign to the members of $S$.
Two prominent such classes are \emph{additively separable} and \emph{fractional} hedonic games (ASHGs and FHGs respectively), on which we focus in this work;
in ASHGs each agent's utility is the sum of the values she assigns to the members of her coalition \cite{BoJa02a}, while in FHGs it is the average (where a zero self-valuation is included in the average).

An important question is: what defines a desirable partition?
As is typically the case in game-theoretic settings there are various ways to define desirable outcomes, and we often seek stable outcomes which are unlikely to fall apart.
The literature offers a variety of stability notions formalizing this intuition, such as Nash stability, individual stability, contractual stability, and core stability.
In this work, we focus on the stability notion called \emph{strong popularity}.
Introduced by \cite{Gard75a} in the context of bipartite matching (which is a subclass of hedonic games), strong popularity aims to capture the principle of decision making by majority vote.
In more detail, a partition is called popular if it weakly beats any other partition in a majority vote among the agents, and it is called strongly popular if it strictly beats any other partition.
In the more general context of voting theory, strongly popular partitions can be seen as the Condorcet winners of hedonic games.
Introduced by \cite{Cond85a}, Condorcet winners are candidates in a voting scheme which win a majority vote among the voters against any other candidate.

A surprising property of Condorcet winners and, indeed, strongly popular partitions, is that they need not exist: majority votes may yield cycles of majority victories among candidates, and thus there does not necessarily exist a candidate that is the ``most popular''. 
This is known as the Condorcet paradox  \cite{Cond85a}.
Since additionally the number of possible partitions in a given hedonic game is exponential in its description size, it is also computationally nontrivial to determine if such a partition exists.
We therefore study this existence problem in the context of ASHGs and FHGs, denoted \ASHG{} and \FHG{}.
\cite{BrBu20a} established \coNP{}-hardness for both problems, and their precise complexity was left as an open question by \cite{BrBu20a,BG24popularity}. 
In this work, we resolve both questions.

In \cite{BG24popularity}, it is shown that determining the existence of (weakly) popular partitions is \Sig{2}-complete in ASHGs and FHGs, where the latter holds even when all agents' valuations are non-negative.
Moreover, in \cite{BrBu20a} it is shown that even verifying whether a given partition is popular or strongly popular is \coNP{}-complete in ASHGs and FHGs.
This hardness of the verification stands in contrast with numerous other stability notions such as Nash, individual, or contractual stability, which can all be immediately included in \NP{}.

A second and arguably more intriguing property of strongly popular partitions, this time not shared with its weak-popularity counterpart, is that they must be unique whenever they exist.
Indeed, no two strongly popular partitions may co-exist as they cannot both beat each other in a majority vote.
The property of having none or one witness for a decision problem is called \emph{unambiguity}.
In the recent work of \cite{GGKN2025unambiguity}, it is noted that unambiguity may affect a problem's computational complexity, as a reduction from an ambiguous problem could not be \emph{parsimonious}, that is, it could not provide a mapping that preserves the number of solutions \cite{hartmanis1976isomorphisms,simon1977difference,valiant1974reduction}.
Notably, this restriction on the applicability of parsimonious reductions may help explain why the complexity of weak popularity was successfully characterized, while that of strong popularity has remained open.

Consequently, several complexity classes associated with unambiguous problems, lying between the first and second levels of the polynomial hierarchy, are introduced by \cite{GGKN2025unambiguity}.
One of those classes is denoted \PCW{}, for Polynomial Condorcet Winner.
In relation to established complexity classes, they show that \PCW{} contains \PNP{}---the class of decision problems solvable in polynomial time given access to an \NP{} oracle (which contains both \NP{} and \coNP{}); and is contained in \SymP{}---the class of symmetric alternations, which roughly captures problems in $\SigPi{2}$ whose yes- and no-witnesses are defined symmetrically.
\SymP{} is shown in \cite{cai2007s2p} to be contained in \ZPPNP{}, the class of problems solvable in probabilistic polynomial time with zero error given access to an \NP{} oracle, which in turn is a subclass of \SigPi{2}.
In the current work, we show that \PCW{} is the correct class to capture strong popularity in ASHGs and FHGs.
For FHGs, our result holds even if all valuations between agents in the game are non-negative.
In contrast, we show that for ASHGs, assuming non-negative valuations diminishes the complexity to \Pclass{}.
Our contributions are summarized by the following theorems (formal problem definitions can be found in \Cref{sec:prelims}).

\begin{theorem}
\label{thm:ASHG_general}
\ASHG{} is \PCW{}-complete.
\end{theorem}

\begin{theorem}
\label{thm:FHG}
\FHG{} is \PCW{}-complete, even if valuations are non-negative.
\end{theorem}

\begin{theorem}
\label{thm:ASHG_nonnegative}
\ASHG{} with non-negative valuations is in \Pclass{}.
Furthermore, when it exists, a strongly popular partition can be found in polynomial time.
\end{theorem}

\Cref{thm:ASHG_general,thm:FHG} are the first natural complete problems for \PCW{}, namely who are not formalized using an arbitrary Boolean circuit, and, to our knowledge, are also the first completeness result for strong popularity in hedonic games.
For \Cref{thm:ASHG_nonnegative}, we show that a strongly popular partition exists if and only if the graph induced by all positive valuations in the game is connected.

\subsection{The Role of Non-Negative Valuations}
An interesting question is why restricting valuations to be non-negative drastically reduces the complexity of strong popularity in ASHGs, while leaving it unchanged in FHGs.
The key observation is that in non-negative ASHGs, the grand coalition (containing all agents) is always (weakly) popular, and is therefore the only candidate for a strongly popular partition.
This considerably simplifies the existence problem: it suffices to determine whether another partition ties with the grand coalition in popularity, in which case no strongly popular partition exists.
Moreover, once existence is established, \emph{finding} a strongly popular partition is trivial: simply output the grand coalition.

In FHGs, however, agents' utilities also depend on the sizes of their coalitions.
Consequently, even with non-negative valuations, agents may prefer smaller coalitions, as adding more agents can decrease their average utility.
Our results show that this additional consideration preserves the computational hardness of strong popularity, even under the non-negativity restriction.

\subsection{On the Class \PCW{}}
\PCW{} is defined as the set of all problems reducible to a problem called \CKTCondorcet{}, in which candidates are represented by length-$n$ bit strings and voters by Boolean circuits. 
Each Boolean circuit takes a candidate as input, and outputs a numerical value (another length-$n$ bit string) representing the voter’s evaluation of that candidate.
This induces a (weak) total preference order for each voter over the different candidates.
It is immediate that this class contains \ASHG{} and \FHG{}, and in fact it contains the decision problem of strong popularity corresponding to \emph{any} cardinal hedonic game, since the output of a circuit corresponds precisely to the cardinal value an agent assigns to a candidate.

Showing \PCW{}-hardness is a much more intricate task.
It requires converting the Boolean gates inside the circuits comprising an instance of \CKTCondorcet{} into agents in a hedonic game in a way that transfers the strong popularity of an $n$-bit string to the strong popularity of some corresponding partition of the agents in the constructed game.
To prove \Cref{thm:ASHG_general,thm:FHG}, our reductions achieve this correspondence.

The remainder of the paper is organized as follows.
In \Cref{sec:Related_Literature} we give an account of related literature, and in \Cref{sec:prelims} we establish necessary notation and definitions. 
In \Cref{sec:proof_sketches} we provide sketches of the proofs of \Cref{thm:ASHG_general,thm:FHG}, highlighting their similarities and differences.
In \Cref{sec:ASHG_general,sec:FHG,sec:ASHG_NN} we prove \Cref{thm:ASHG_general,thm:FHG,thm:ASHG_nonnegative} respectively.

\section{Related Literature}
\label{sec:Related_Literature}
Early work on hedonic games includes \cite{DrGr80a}, followed by \cite{BoJa02a,BKS01a,CeRo01a}.
The book chapters of \cite{AzSa15a} and \cite{BER24a} provide an overview of the topic.

In a general model of hedonic games, agents need to express preferences over exponentially many coalitions, making it difficult to represent the game concisely.
To overcome this, succinct classes of hedonic games have been proposed, many of which are based on cardinal valuations, as in ASHGs, FHGs, and modified fractional hedonic games \cite{Olse12a}.
Related hedonic-game settings include roommates and flatmates games in which coalitions have size at most two and three, respectively \cite{aziz2013pareto,GMSZ21a,BIM10a,BrBu20a}, and friends-enemies (or friends-enemies-neutrals) games where valuations are in $\{-1,1\}$ (or in $\{-1,0,1\}$) \cite{chen2023hedonic,ohta2017core,dimitrov2006simple}.
\cite{bounded_ASHGs2026} further study ASHGs with various bounds on coalition sizes.

As briefly mentioned in \Cref{sec:introduction}, various solution concepts aiming to describe stable partitions have been proposed in the literature.
Interestingly, there is a typical dichotomy based on whether stability is defined with respect to individual or group deviations.
Stability notions based on single-agent deviations often lead to \NP{}-completeness. 
For instance, in ASHGs the existence problems corresponding to Nash-stability, individual stability, and contractual stability are \NP{}-complete \cite{SuDi10a,BBT23a}, and in FHGs the same is shown for Nash stability and individual stability \cite{BBS14a}.
In various settings of friends-enemies-neutrals games, \NP{}-completeness has also been established for \emph{local stability} and \emph{local popularity}, two variants of popularity based on single-agent deviations \cite{csaji2026clustering}.
In contrast, in ASHGs the existence problems corresponding to the core, strict core, and weak popularity are \Sig{2}-complete \cite{SuDi10a,ABS11c,Woeg13a,BG24popularity}, and in FHGs the same is shown for the core and weak popularity \cite{ABB+17a,BBS14a,BG24popularity}.

While we presently show that strong popularity yields a complexity characterization different from those of the above solution concepts, it is interesting to note that in ASHGs and FHGs, \emph{verifying} that a given partition is stable is \coNP{}-complete with respect to both weak and strong popularity \cite{BrBu20a}.

The existence of strong popularity has been previously studied.
\cite{BrBu20a} showed that, in ASHGs and FHGs, this problem is \coNP{}-hard, though completeness was out of reach without the proper framework to handle its unambiguity---the property of having at most one witness.
\cite{GGKN2025unambiguity} study unambiguous problems within \Sig{2}. 
They introduce several complexity classes based on different combinatorial criteria that induce unambiguity.
Apart from the class \PCW{}, they define a more general class termed \PTW{} (Polynomial Tournament Winner), based on the problem of determining the existence of a source in a tournament; and another class \PMA{} (Polynomial Majority Argument) based on the principle that at most one entity can claim a majority of some shared commodity.
All three classes are contained in \SymP{}, indicating a significant gap from \Sig{2}.

\section{Preliminaries}
\label{sec:prelims}
Let $N$ be a set of agents. 
A \emph{coalition} is a non-empty subset of $N$. 
A \emph{singleton} coalition is a coalition of size one, and the \emph{grand} coalition is simply $N$.
Let $\mathcal{N}_i=\{S\subseteq N\colon i\in S\}$ denote the set of all coalitions to which agent $i$ belongs. 
A \textit{coalition structure} (or a \textit{partition}) is a collection of disjoint nonempty coalitions whose union is $N$. 
For an agent $i\in N$, we denote by $\pi(i)$ the coalition containing $i$ in a partition $\pi$. 

A \textit{hedonic game} is a pair $\langle N,\succsim \rangle$, where ${\succsim}={(\succsim_i)}_{i\in N}$ is a preference profile specifying for each agent $i$ their preferences as a complete and transitive preference order $\succsim_i$ over $\mathcal{N}_i$. 
In hedonic games, agents' preferences are only affected by the members of their own coalition, and therefore $\succsim_i$ induces an underlying preference order over partitions as well, given by $\pi\succsim_i\pi'$ if and only if $\pi(i)\succsim_i\pi'(i)$.
For coalitions $S,S'\in \mathcal{N}_i$, we say that agent $i$ \textit{weakly prefers} $S$ over $S'$ if $S\succsim_i S'$.
Moreover, we say that $i$ \textit{prefers} $S$ over $S'$ if $S\succ_i S'$.
We use the same terminology for preferences over partitions.

Although agents' preferences may be arbitrary in general, in this paper we focus on \emph{additively separable} and \emph{fractional} hedonic games (ASHGs and FHGs), following \cite{BoJa02a,ABB+17a}. 
Those models are specified by a pair $\langle N,\mathbf{v} \rangle$ where $N$ is a set of agents, and $\mathbf{v} = (\mathbf{v}_i\colon N\to \RR)_{i\in N}$ is a collection of \textit{valuation functions}. 
The quantity $\mathbf{v}_i(j)$ denotes the value agent $i$ assigns to agent $j$. 
We assume that $\mathbf{v}_i(j)$ is a rational number, for ease of representation.
In ASHGs, we define the utility of agent $i$ in coalition $S$ by $u_i(S)=\sum_{j\in S\bs\{i\}}\mathbf{v}_i(j)$, namely the sum of values she assigns to the other members of her coalition. 
In FHGs, we define the utility of agent $i$ in coalition $S$ by $u_i(S)=\frac{\sum_{j\in S\bs\{i\}}\mathbf{v}_i(j)}{|S|}$, which can be interpreted as the average of values she assigns to the members of her coalition (when implicitly including a value $0$ to herself).  
We then define the preference orders $\succsim$ of the agents by $S \succ_i S'$ if and only if $u_i(S)>u_i(S')$ and $S \succsim_i S'$ if and only if $u_i(S)\ge u_i(S')$, for any two coalitions $S$ and $S'$.
We extend the utilities for partitions by setting $u_i(\pi)=u_i(\pi(i))$, for any agent $i$ and partition $\pi$.

Given two partitions $\pi$ and $\pi'$ and a coalition $S\subseteq N$, we define the \emph{popularity margin} of $S$ on the ordered pair $(\pi,\pi')$ by
\[\phi_S(\pi,\pi')=|\{a\in S\colon u_a(\pi)>u_a(\pi')\}|-|\{a\in S\colon u_a(\pi')>u_a(\pi)\}|.\]
Note that in this definition, agents who are indifferent between the two partitions do not contribute to any of the two terms.
For the grand coalition we write $\phi(\pi,\pi')=\phi_N(\pi,\pi')$, and for a singleton coalition containing only one agent $\ga$ we write $\phi_{\ga}(\pi,\pi')=\phi_{\{\ga\}}(\pi,\pi')$.
The definition of popularity margins is useful as it is often convenient to consider restricted subsets of agents, and then aggregate their margins (since, if $S$ and $S'$ are disjoint coalitions, then their popularity margins are additive, namely $\phi_S(\pi,\pi')+\phi_{S'}(\pi,\pi')=\phi_{S\cup S'}(\pi,\pi')$).
We say that a partition $\pi$ is \emph{more popular} than $\pi'$ if $\phi(\pi,\pi')>0$. 
We say $\pi$ is \emph{strongly popular} if $\pi$ is more popular than any partition $\pi'\ne\pi$.

\begin{wbox}
\label{def:FHG_strong_popularity}
\ASHG{} (resp. \FHG)\\
\textbf{Input:} Additively separable (resp. fractional) hedonic game $\langle N,\mathbf{v} \rangle$.\\
\textbf{Question:} Does $\langle N,\mathbf{v} \rangle$ admit a strongly popular partition? Namely, does the following hold:
\[\exists\, \text{partition }\pi^* \;\; \forall\, \text{partition }\pi \text{ with } \pi\neq \pi^*\;\; \phi(\pi^*,\pi)>0\]
\end{wbox}

In the context of popularity it is useful to discuss Pareto optimality. 
We say that $\pi'$ is a \textit{Pareto improvement} from $\pi$ if all agents weakly prefer $\pi'$ over $\pi$, and at least one agent strictly prefers $\pi'$ over $\pi$. 
If there exists no Pareto improvement from $\pi$, we say $\pi$ is \textit{Pareto-optimal}. 
Clearly, strongly popular partitions are Pareto-optimal.
Indeed, every Pareto improvement is a more popular partition.
By contrast, Pareto-optimal partitions need not be strongly popular.

Next, we include a formal definition of the class \PCW{}.

\begin{wbox}
    \CKTCondorcet\\
    \textbf{Input:} A sequence of Boolean circuits $\fCC=\langle\CC_1,\dots,\CC_m\rangle$, each with $n$ inputs and $n$ outputs.\\
    \textbf{Question:}
    $\exists \XX\in\bits{n} \;\; \forall \YY\in\bits{n} \text{ with } \YY\neq \XX\;\; \sum_{i=1}^m\big(\sgn(\CC_i(\XX)-\CC_i(\YY))\big)>0$?\\
    A solution is called a \emph{Condorcet string}.
\end{wbox}

\begin{definition}
(\cite{GGKN2025unambiguity}).
The class Polynomial Condorcet Winner (\PCW{}) is defined by
\[\PCW{}=\{L\subseteq\Sigma^*\colon L\reduce \CKTCondorcet{}\}\]
where $\Sigma^*$ is the set of all finite strings over a finite alphabet $\Sigma$.
\end{definition}

To prove \Cref{thm:ASHG_general,thm:FHG} we reduce to and from \CKTCondorcet{}.
As mentioned in \Cref{sec:introduction}, containment in \PCW{} is straightforward: each partition can be uniquely encoded by a binary string of polynomial length, and the validity of such an encoding can be checked in polynomial time. 
Moreover, since valuations are rational numbers encoded in binary, agents' utilities can be represented by positive rational numbers of polynomial bit length in a way that preserves their preference orders (we may add a constant to all outputs if needed). 
We can therefore construct, for each agent, a Boolean circuit that receives an encoding of a partition and outputs the corresponding utility, or $0$ if the encoding is invalid. 
It follows that strongly popular partitions of the hedonic game correspond precisely to Condorcet strings of the constructed \CKTCondorcet{} instance.
Thus, in \Cref{sec:ASHG_general,sec:FHG} we only prove \PCW{}-hardness of the corresponding problems.

We briefly introduce some notation concerning Boolean circuits. Let $\CC$ be a Boolean circuit with $n$ input bits, let $\GG$ be a gate in $\CC$, and let $\XX = (\XX_i)_{i\in [n]} \in\{0,1\}^n$. We denote by $\XX(\GG)$ the value ($0$ or $1$) that the gate $\GG$ outputs when the circuit $\CC$ is given the string $\XX$ as input. Moreover, we denote by $|\CC|$ the number of gates in $\CC$.

Lastly, we describe some assumptions on the structure of \CKTCondorcet{} instances.
Let $\fCC=\langle\CC_1,\dots,\CC_m\rangle$ be a \CKTCondorcet{} instance, where each circuit $\CC_j$ is composed of gates $\GG_1^{j},\dots,\GG_{|\CC_j|}^{j}$.
Since we usually focus on an individual circuit in our analysis, we often omit the superscripts of the gates, writing, e.g., $\GG_i$ for $\GG_i^j$.
We assume without loss of generality that $m\ge 3$ (otherwise, add dummy circuits which always output $0$).
We also assume without loss of generality that the outputs of all circuits are interpreted as non-negative binary numbers, as otherwise we can add a sufficiently large constant to all outputs without affecting the preference order of the circuits.
Furthermore, we assume that the gates are topologically ordered, so that, in a specific circuit, if $\GG_i$ is an input to $\GG_j$ then $i<j$.
Moreover, we assume that the first $n$ gates of each circuit $\GG_1,\dots,\GG_n$ simply copy the inputs $\XX_1,\dots,\XX_n$ respectively, and these are the only gates where $\XX_1,\dots,\XX_n$ appear; we call those gates \emph{Copy-gates}. Apart from those, we assume circuits are only composed of And- and Not-gates. 
The output gates of each circuit are simply some subset of size $n$ of its gates; we also denote them by $\OO_1,\dots,\OO_n$,
in reverse order of significance (where $\OO_n$ is the most significant bit).

\section{Proof Sketch for Main Results}
\label{sec:proof_sketches}
The constructions in the proofs of \Cref{thm:ASHG_general,thm:FHG} follow a similar high-level approach, while the technical details differ significantly due to the inherent differences between ASHGs and FHGs and the additional non-negativity constraint in \Cref{thm:FHG}. 
In this section we provide a sketch for both proofs, highlighting similarities and differences between them.

\subsection{Overview of Both Reductions}
In both constructions, separate groups of agents of the constructed game capture the essence of some component of a given \CKTCondorcet{} instance $\fCC=\langle\CC_1,\dots,\CC_m\rangle$.

The main components are:
\begin{enumerate}
    \item $X$-agents representing the input bits of the circuits. All $X$-agents are strongly drawn to each other, and dictate that any partition which separates them from each other is very unpopular. We call their coalition the \emph{main coalition}.
    \item Assignment gadgets, one for each input bit, representing the value ($0$ or $1$) that each bit is assigned. 
    \item Gate gadgets, one for each gate, representing the value ($0$ or $1$) that the gate obtains.
    \item Voter agents, one for each circuit, representing the value the circuit outputs.
    $X$-agents also value voters strongly, so that all voters must also be in the main coalition to avoid unpopular partitions.
\end{enumerate}
The central idea is that each assignment gadget and each gate gadget has two \emph{core agents} corresponding to the values $0$ and $1$; their coalitions are used to determine the assignment to the corresponding bit, or the value obtained by the corresponding gate.
More specifically, whenever the agent representing value $1$ is in the main coalition while the $0$-agent is not, we interpret the value of the bit/gate as $1$; and whenever the converse holds we interpret the value as $0$.

Another important mechanism in the proofs is the use of \emph{alternative agents}. 
Intuitively, we give each agent not joining the main coalition an equally attractive alternative, so that agents outside the main coalition have a unique reasonable coalition to join.
This is essential firstly because strongly popular partitions are unique whenever they exist, and secondly because we want those agents to be indifferent to the value their corresponding gadget obtains.

With these components in mind, we show that the \emph{natural partitions}, where each gate and assignment gadget is partitioned in a way that reflects the value it should obtain with respect to its inputs' assignments, are the only candidates for strongly popular partitions, as they beat any non-natural partition in a majority vote.
We additionally show that all non-voter agents are indifferent between the natural partitions.
Thus, we need only consider the utilities obtained by the voter agents.
With a careful choice of valuations for the voters, we show that their utilities for a given natural partition reflect the values that the circuits would obtain when evaluated on the input string corresponding to that partition.
This completes the argument, as a strongly popular partition translates to a Condorcet string and vice versa.

\textbf{Some notes on the technical details of the constructions:}
As described, all non-voter agents should capture the inner structure of the circuits while maintaining the required connections to the input bits, while voters should only reflect the eventual output values of the circuits.
To achieve this, there are several obstacles our constructions address:
\begin{itemize}[leftmargin=*]
    \item Voters should not be too influential, so that they cannot force an unnatural partition that does not correspond to any assignment.
    For this reason, we introduce \emph{replica agents} in each gadget, which intuitively amplify an agent's influence in majority comparisons. 
    An unnatural partition may potentially benefit all $m$ voters, but it would hurt at least one agent along with all her replicas, rendering it less popular than a natural partition. 
    \item Non-voter agents should be indifferent towards different natural partitions, so that the popularity of the partitions is measured purely by the voters, reflecting the circuits' preferences.
    For most agents, this is not too difficult to achieve.
    For And-gate gadgets, however, we use special agents (denoted $\Wand_i$, $\Zand_i$, and $\Yand_i$ in the constructions), and a careful choice of utilities to achieve this.
    \item Voters should prefer one partition over another purely by the popularity margin of the circuits' preferences over the corresponding bit-string assignment.    
    We achieve this differently in each construction, as described in \Cref{sec:sketch_differences} below.\label{sketch_item:voters_preferences}
\end{itemize}

\subsection{Differences Between the Reductions}
\label{sec:sketch_differences}
There are several crucial differences in the reductions.
In FHGs the size of the coalition affects agents' utilities. 
Thus, it is harder to model neutrality between agents. 
In ASHGs, we can use $0$ as a neutral value, as no agent prefers to include or exclude from their coalition another agent to whom they assign $0$.
In FHGs, no fixed value is universally neutral, as it depends on context.
To overcome this, in the FHG construction our aim is to have all agents (apart from voters) obtain utility $1$ in any natural partition; thus, the value $1$ itself becomes neutral (since, if $\ga$ assigns value $1$ to $\ga'$, and $\ga$ obtains utility $1$ in her coalition, then adding $\ga'$ has no effect on the utility of $\ga$), and coalition sizes become irrelevant.

Having agents obtain utility $1$ in any natural partition helps overcome another obstacle, which is the non-negativity assumption of the FHG proof.
This assumption makes it harder to separate agents from each other: while in ASHGs we can easily make certain collaborations between agents consistently inefficient (regardless of the composition of the coalition) by having them assign large negative values to each other, in non-negative FHGs we cannot do that.
There, the lowest value agents can assign to each other is $0$; if an agent $\ga$ assigns $0$ to another member $\ga'$ of her coalition, $\ga$ is not necessarily worse off than she would be in some other coalition with no such member $\ga'$.
However, since agents can generally obtain utility $1$ in a natural partition, we can use $0$ as a ``bad'' value, as adding an agent to whom you assign $0$ would decrease the agent's utility.

The fact that in FHGs utilities depend on coalition sizes also affects the treatment of the voters (see \Cref{sketch_item:voters_preferences} above).
In particular, since voters are designed to capture the output values of the circuit, we do not wish for them to have a fixed utility of $1$ in any natural partition like the other agents.
This makes it harder to make them neutral to other agents, compared with the ASHG construction, where voters simply use the $0$ value for neutrality.
In the FHG proof, our solution for this is to force the size of the voters' coalition (that is, the main coalition) to be the same in any natural partition.
Thus, in the eyes of a voter, all that changes between each natural partition is the value obtained by her corresponding circuit.
To achieve this uniformity across all natural partitions, we must ensure that each gadget contributes the same number of agents to the main coalition regardless of the value the gate obtains in the corresponding assignment.
For most gate types this is easy by symmetry of the gadgets, but the And-gate gadgets require more careful attention (this is the purpose of the increased number of replicas of the agent $\gt_i$ corresponding to the value $1$ of an And-gadget, compared with other agents).


\section{Strong Popularity in Arbitrary ASHGs}
\label{sec:ASHG_general}
In this section we prove \Cref{thm:ASHG_general}.
In \Cref{sec_ASHG:setup} we describe the reduction from \CKTCondorcet{}, and in \Cref{sec_ASHG:correctness} we prove its correctness.
We introduce a useful definition to help us construct an ASHG for the reduction.
\begin{definition}
\label{def_ASHG:replica-family}
\textbf{(ASHG One-way replicas).}
Let $\langle N,\mathbf{v}\rangle$ be an ASHG, let $\ga\in N$ and let
$S\subseteq N\bs\{\ga\}$ be a set of distinct agents.
We say that $S$ is a set of \emph{one-way replicas} (or simply \emph{replicas}) of the \emph{origin} $\ga$, if the following hold:
\begin{enumerate}
  \item For every distinct $\ga',\ga''\in S\cup\{\ga\}$, we have $\mathbf{v}_{\ga'}(\ga'')=\mathbf{v}_{\ga''}(\ga')=10$.
  \item For every $\ga'\in N\bs(S\cup \{\ga\})$ and every $r\in S$, we have:
  \begin{enumerate}
      \item $\mathbf{v}_{r}(\ga')=\mathbf{v}_{\ga}(\ga')$, and
      \item if $\mathbf{v}_{\ga'}(\ga)\ge 0$ then $\mathbf{v}_{\ga'}(r)=0$, and if $\mathbf{v}_{\ga'}(\ga)<0$ then $\mathbf{v}_{\ga'}(r)=\mathbf{v}_{\ga'}(\ga)$.
  \end{enumerate}
\end{enumerate}
The set of replicas of $p$ is also denoted $r(p)=S$.
\end{definition}
Notice that in this definition $\ga$ and its replicas assign the same valuation to other agents, and, intuitively, they have a strong incentive to belong to the same coalition. 
The mutual value of~$10$ is specifically tailored to our reduction and designed to be greater than the sum of all other positive outgoing valuations of the replicated agents in the constructed game. Furthermore, any coalition whose members like the presence of $\ga$ (i.e. have a non-negative valuation for $\ga$) will not object to including its replicas as well, though including the replicas will not affect the utilities of other agents in this coalition (apart from $\ga$ and its replicas). 
Thus, one may interpret a replica as a ``weight multiplier'', i.e., a tool that enables us to add more weight to the opinion of agent $\ga$.

Containment of \ASHG{} in \PCW{} is straightforward, as each agent can be easily encoded as a Boolean circuit calculating the agent's utility for a given partition.
Thus, to prove \Cref{thm:ASHG_general} we only prove \PCW{}-hardness.

\subsection{Setup of the Reduction}
\label{sec_ASHG:setup}
Suppose that we are given a \CKTCondorcet{} instance  $\fCC=\langle\CC_1,\dots,\CC_m\rangle$, with properties as described in \Cref{sec:prelims}.
We construct the following ASHG $\langle N,\mathbf{v}\rangle$. 
We let $N=X\cup A\cup G\cup L\cup W\cup Z\cup V$, where those disjoint sets induce different \emph{types} of agents denoted \emph{$T$-agents} for $T\in\{X,A,G,L,W,Z,V\}$, as detailed below. 
For $L$-agents, we make a further distinction between four sub-types, so that $L=\La\cup\Lg\cup\Lw\cup\Lz$.
Note that all one-way replicas in the construction are assumed to belong to the same type as their origin agent.
\begin{itemize}
    \item We create $X$-agents $x_1,\dots,x_n$, representing the input bits $\XX_1,\dots,\XX_n$ respectively. Furthermore, we create additional $X$-agents (not corresponding to any input bit) to increase the overall number of $X$-agents to $1+9mn+3n+\sum_{i=1}^m\big(7(m+1)|\CC_i|+1\big)$. 
    \item For each input bit $\XX_i$ we create two \emph{assignment agents} (or $A$-agents) $\asst_i$ and $\assf_i$, and an \emph{assignment-alternative agent} (or $\La$-agent) $\ls_i$.
    Furthermore, for each of those three agents we create $m$ one-way replicas. We denote $\Ta_i=\{\asst_i\}\cup r(\asst_i)$, $\Fa_i=\{\assf_i\}\cup r(\assf_i)$, and $\La_i=\{\ls_i\}\cup r(\ls_i)$.
    We will use the coalitions of the $\asst_i$ and $\assf_i$ agents to derive an assignment to the input string, and vice versa.
    \item For each gate $\GG_i$ (Copy-, Not-, or And-gate), we create two \emph{gate agents} (or $G$-agents) $\gt_i$ and $\gf_i$, and a \emph{gate-alternative agent} (or \emph{$\Lg$-agent}) $\lg_i$.
    Furthermore, for each of those three agents we create $m$ one-way replicas. We denote $\Tg_i=\{\gt_i\}\cup r(\gt_i)$, $\Fg_i=\{\gf_i\}\cup r(\gf_i)$, and $\Lg_i=\{\lg_i\}\cup r(\lg_i)$.
    We will use the coalitions of the $\gt_i$ and $\gf_i$ agents to derive the outputs of the gates, and vice versa.
    \item For each And-gate $\GG_i$ we additionally construct a \emph{$\Wand$-agent} $\wand_i$ and a \emph{$\Zand$-agent} $\zand_i$. 
    We further construct an $\Lw$-agent $\lw_i$ and an $\Lz$-agent $\lz_i$.
    Moreover, for each of $\{\wand_i,\zand_i,\lw_i,\lz_i\}$ we construct $m$ one-way replicas. 
    We denote $\Wand_i=\{\wand_i\}\cup r(\wand_i)$, $\Zand_i=\{\zand_i\}\cup r(\zand_i)$, $\Lw_i=\{\lw_i\}\cup r(\lw_i)$, and $\Lz_i=\{\lz_i\}\cup r(\lz_i)$.
    \item For each circuit $\CC_j$, we create a \emph{voter agent} (or \emph{$V$-agent}), denoted $v_j$. When considering a specific circuit we may omit the subscript of this agent. $v$ should not be confused with the valuation function $\mathbf{v}$.
\end{itemize}

For each input bit $\XX_i$, the set $\Ta_i\cup\Fa_i\cup\La_i$ is called the \emph{assignment gadget}, or \emph{$A$-gadget}, of $\XX_i$, denoted $\ADD_i$.
For each gate $\GG_i$, the set $\Tg_i\cup\Fg_i\cup\Lg_i$ (or, if $\GG_i$ is an And-gate, then $\Tg_i\cup\Fg_i\cup\Lg_i\cup\Wand_i\cup\Zand_i\cup \Lw_i\cup \Lz_i$) is called the \emph{gate gadget} of $\GG_i$, denoted $\GDD_i$.
We sometimes refer to them as Copy-, Not-, or And-gadgets when referring to a specific type of gate. 
The $X$-, $A$-, $V$-, and $G$-agents are referred to as \emph{core agents}. 

Given a circuit $\CC_j$, we denote by $C_j$ its corresponding set of agents (which includes the gadgets of all gates of $\CC_j$, and its voter agent).
The number of non-$X$-agents is at most $9mn+3n+\sum_{i=1}^m\big(7(m+1)|\CC_i|+1\big)$, where we bound with the case that all gates are Copy- and And-gates.
Hence, we have that $|X|>\frac{|N|}{2}$, which will be helpful during the proof.

We now describe the valuations $\mathbf{v} = (\mathbf{v}_i\colon N\to \RR)_{i\in N}$ of the agents. 
Any directed valuation not specified below is zero.
We write $-\infty$ as shorthand for a sufficiently large negative number, satisfying the property that an agent who assigns $-\infty$ to an agent in her coalition will have a negative utility regardless of the other members of her coalition. For instance, we can regard $\infty$ as denoting $|N|$ times the largest (positive) valuation in the game.
In the following description we do not include the valuations of all one-way replicas, since those can be derived from \Cref{def_ASHG:replica-family}, using the valuations of their origin agents.
\begin{itemize}
    \item Each $X$-agent $x$ assigns $1$ to any other agent in $X\cup V$.
    \item For each input bit $\XX_i$:
    \begin{itemize}
        \item agents $\asst_i$ and $\assf_i$ assign value $-\infty$ to each other, and value $1$ to $x_i$ and to $\ls_i$;
        \item agent $\ls_i$ assigns value $1$ to agents $\asst_i$ and $\assf_i$.
    \end{itemize}
    \item For every gate $\GG_i$:
    \begin{itemize}
        \item agents $\gt_i$ and $\gf_i$ assign value $-\infty$ to each other;
        \item agent $\lg_i$ assigns value $1$ to agents $\gt_i$ and $\gf_i$.
    \end{itemize}
    \item For every Copy-gate $\GG_i=\XX_i$: 
    \begin{itemize}
        \item agent $\gt_i$ assigns value $1$ to agents $\asst_i$ and $\lg_i$, and $-\infty$ to $\assf_i$;
        \item agent $\gf_i$ assigns value $1$ to agents $\assf_i$ and $\lg_i$, and $-\infty$ to $\asst_i$;
        \item agent $\asst_i$ assigns $-\infty$ to $\gf_i$.
        \item agent $\assf_i$ assigns $-\infty$ to $\gt_i$.
    \end{itemize}
    \item For every Not-gate $\GG_i=\neg \GG_j$:
    \begin{itemize}
        \item agent $\gt_i$ assigns value $1$ to agents $\gf_j$ and $\lg_i$, and $-\infty$ to $\gt_j$;
        \item agent $\gf_i$ assigns value $1$ to agents $\gt_j$ and $\lg_i$, and $-\infty$ to $\gf_j$;
        \item agent $\gt_j$ assigns $-\infty$ to $\gt_i$;
        \item agent $\gf_j$ assigns $-\infty$ to $\gf_i$.
    \end{itemize}
    \item For every And-gate $\GG_i=\GG_j\land \GG_k$:
    \begin{itemize}
        \item agent $\gt_i$ assigns $1$ to $\gt_j$ and to $\gt_k$, $2$ to $\lg_i$, and $-\infty$ to $\gf_j$, $\gf_k$, $\wand_i$, and $\zand_i$.
        \item agent $\gf_i$ assigns $1$ to $\gf_j$ and to $\wand_i$, $2$ to $\gf_k$ and to $\zand_i$, and $3$ to $\lg_i$; 
        \item agent $\wand_i$ assigns $-\infty$ to $\gt_i$, $\gt_k$, $\gf_j$, and $\zand_i$, and $1$ to $\lw_i$ and to $\gf_i$;
        \item agent $\zand_i$ assigns $-\infty$ to $\gt_i$, $\gt_j$, $\gf_k$, and $\wand_i$, and $1$ to $\lz_i$ and to $\gf_i$.
        \item agent $\gt_j$ assigns $-\infty$ to $\zand_i$;
        \item agent $\gt_k$ assigns $-\infty$ to $\wand_i$;
        \item agent $\gf_j$ assigns $-\infty$ to $\gt_i$, and $\wand_i$;
        \item agent $\gf_k$ assigns $-\infty$ to $\gt_i$, and $\zand_i$;
        \item agent $\lw_i$ assigns $0$ to $\wand_i$;
        \item agent $\lz_i$ assigns $0$ to $\zand_i$.
    \end{itemize}
    \item Any $L$-agent assigns $-\infty$ to all agents and vice versa, apart from the cases where other values were specified above.
    \item Consider a voter agent $v$ corresponding to circuit $\CC$ with outputs $\OO_1,...,\OO_n$ (where $\OO_n$ is the most significant bit). 
    Then for every $\OO_j,\; j\in[n]$, with corresponding $G$-agents $\ot_j$ and $\of_j$, agent $v$ assigns value $2^{j}$ to $\ot_j$.
\end{itemize}

Note that the $A$-agents are not associated with any specific circuit, but can be interpreted as a barrier between the inputs and the circuits. This helps guarantee that circuits follow the same assignment to the input string.
The role of the alternative agents is to provide an alternative coalition for each gate agent or assignment agent representing the opposite truth value of the corresponding bit or gate (whereas, as we will see, agents who correspond to their gate's/bit's truth assignment will all be grouped in a single main coalition). 
Notice that the value that a $V$-agent assigns to a corresponding output-gate agent $\ot_j$ dominates the sum of all values she assigns to agents associated with less significant bits of the output. This will be crucial to establish that the strong popularity property of an input string is preserved in the corresponding partition in the constructed ASHG, and vice versa.
In the following two subsections, we show that a Condorcet string exists in $\fCC$ if and only if a strongly popular partition exists in the constructed ASHG.
Illustrations of the gadgets and overview of the reduction can be seen in \Cref{fig:ASHG_gadgets,fig:Overview} respectively.

Before proceeding with the proof, we introduce two important definitions.
The first formalizes the correspondence between partitions of gadgets and the values the gate obtains; this is derived from the coalitions of the core agents of the gadget.
The second formalizes the idea of having ``well-behaved'' gate gadgets, in the sense that the value that they obtain (according to the former definition) is the correct one with respect to the values of their input gadgets.

\begin{definition}
\label{def_ASHG:valid_partition}
{\bf (ASHG $\pi$-correspondence, $\pi$-validity).}
Let $\XX_i$ be an input bit with corresponding $A$-gadget $\ADD_i$, let $\pi$ be a partition of $N$, and let $MC:=\pi(x_1)$. 
If $\asst_i\in\MC$ and $\assf_i\in\pi(\ls_i)$ then we say that $\ADD_i$ \emph{$\pi$-corresponds to assignment $\XX_i=1$}, 
and if $\assf_i\in\MC$ and $\asst_i\in \pi(\ls_i)$ then we say that $\ADD_i$ \emph{$\pi$-corresponds to assignment $\XX_i=0$}.
If $\ADD_i$ $\pi$-corresponds to either $0$ or $1$, we say it is a \emph{$\pi$-valid gadget}.

Let $\GG_i$ be a gate. If $\gt_i\in\MC$ and $\gf_i\in\pi(\lg_i)$, we say that $\GDD_i$ \emph{$\pi$-corresponds to value $\GG_i=1$}, and if $\gf_i\in\MC$ and $\gt_i\in\pi(\lg_i)$, we say that $\GDD_i$ \emph{$\pi$-corresponds to value $\GG_i=0$}.
If $\GDD_i$ $\pi$-corresponds to either $0$ or $1$, we say it is a \emph{$\pi$-valid gadget}.
\end{definition}

\begin{definition}
\label{def_ASHG:compliance}
{\bf (ASHG $\pi$-compliance).}
Let $\GG_i$ be a gate with corresponding $G$-gadget $\GDD_i$, and let $\pi$ be a partition of $N$. Suppose the inputs (or input) of $\GG_i$ have $\pi$-valid corresponding gadgets.
We say that $\GDD_i$ \emph{$\pi$-complies with its inputs} if $\GDD_i$ $\pi$-corresponds to the value that $\GG_i$ obtains when evaluated on the values its input gadgets $\pi$-correspond to (e.g., if $\GG_i=\neg\GG_j$ and $\GDD_j$ $\pi$-corresponds to $\GG_j=1$, then $\GDD_i$ $\pi$-complies with its input if it $\pi$-corresponds to $\GG_i=0$).

Furthermore, given a string $\XX\in\{0,1\}^n$ we say that $\GDD_i$ \emph{$\pi$-complies with $\XX$} if $\GDD_i$ $\pi$-corresponds to the value that $\GG_i$ obtains when its circuit is evaluated on $\XX$.
\end{definition}

\begin{figure}[ht]
\centering
\begin{minipage}{0.45\textwidth}
    \centering
    \begin{tikzpicture}[
    every node/.style={inner sep=2pt},
    >=stealth
]

\node[circle, fill, minimum size=3pt, inner sep=0pt, label=above:$x_i$] (x) {};

\node[circle, fill, minimum size=3pt, inner sep=0pt, label=below left:$\assf_i$] (assf)   [below left=1.2cm and 0.6cm of x] {};
\node[circle, fill, minimum size=3pt, inner sep=0pt, label=below right:$\asst_i$] (asst)   [below right=1.2cm and 0.6cm of x] {};

\node[color=blue, circle, draw, minimum size=3.5mm, inner sep=0pt, label={[text=gray]below:$m$}] (Fa)  [left=1cm of assf] {};
\node[color=blue, circle, draw, minimum size=3.5mm, inner sep=0pt, label={[text=gray]below:$m$}] (Ta)  [right=1cm of asst] {};

\node[circle, fill, minimum size=3pt, inner sep=0pt, label=below left:$\ls_i$] (la)   [below=1cm of $(assf)!0.5!(asst)$] {};

\node[color=blue, circle, draw, minimum size=3.5mm, inner sep=0pt, label={[text=gray]right:$m$}] (La)  [below=1cm of la] {};

\draw[red] (asst) -- (assf);
\draw [gray, -, shorten >=0cm, shorten <=0.3cm] (asst) -- (x);
\draw [gray, ->, shorten >=1cm, shorten <=0cm] (asst) -- (x);
\draw [gray, -, shorten >=0cm, shorten <=0.3cm] (assf) -- (x);
\draw [gray, ->, shorten >=1cm, shorten <=0cm] (assf) -- (x);
\draw [gray](asst) -- (la);
\draw [gray](assf) -- (la);
\draw[blue, thick] (asst) -- (Ta);
\draw[blue, thick] (assf) -- (Fa);
\draw[blue, thick] (la) -- (La);

\end{tikzpicture}
    \caption{Assignment gadget of bit $\XX_i$.}
    \label{ASHG_fig:A-gadget}
\end{minipage}
\hfill
\begin{minipage}{0.45\textwidth}
    \centering
    \begin{tikzpicture}[
    every node/.style={inner sep=2pt},
    >=stealth
]

\node[circle, fill, minimum size=3pt, inner sep=0pt, label=left:$\assf_i$] (assf) {};
\node[circle, fill, minimum size=3pt, inner sep=0pt, label=right:$\asst_i$] (asst) [right=1.2cm of assf] {};

\node[circle, fill, minimum size=3pt, inner sep=0pt, label=below left:$\gf_i$] (gf)   [below=1.3cm of assf] {};
\node[circle, fill, minimum size=3pt, inner sep=0pt, label=below right:$\gt_i$] (gt)   [below=1.3cm of asst] {};

\node[color=blue, circle, draw, minimum size=3.5mm, inner sep=0pt, label={[text=gray]below:$m$}] (ggf)  [left=1cm of gf] {};

\node[color=blue, circle, draw, minimum size=3.5mm, inner sep=0pt, label={[text=gray]below:$m$}] (ggt)  [right=1cm of gt] {};

\node[circle, fill, minimum size=3pt, inner sep=0pt, label=below left:$\lg_i$] (lg)   [below=1cm of $(gf)!0.5!(gt)$] {};

\node[color=blue, circle, draw, minimum size=3.5mm, inner sep=0pt, label={[text=gray]right:$m$}] (lgg) [below=1cm of lg] {};

\draw[red] (gt) -- (gf);
\draw[red] (gt) -- (assf);
\draw[red] (gf) -- (asst);
\draw [gray, -, shorten >=0cm, shorten <=0.3cm] (gt) -- (asst);
\draw [gray, ->, shorten >=1cm, shorten <=0cm] (gt) -- (asst);
\draw [gray] (gt) -- (lg);
\draw [gray, -, shorten >=0cm, shorten <=0.3cm] (gf) -- (assf);
\draw [gray, ->, shorten >=1cm, shorten <=0cm] (gf) -- (assf);
\draw [gray](gf) -- (lg);
\draw[blue, thick] (gt) -- (ggt);
\draw[blue, thick] (gf) -- (ggf);
\draw[blue, thick] (lg) -- (lgg);

\end{tikzpicture}
    \caption{Copy-gadget of gate $\GG_i=\XX_i$.}
    \label{ASHG_fig:Copy-gadget}
\end{minipage}

\vspace{1em} 

\begin{minipage}{0.45\textwidth}
    \centering
    \begin{tikzpicture}[
    every node/.style={inner sep=2pt},
    >=stealth
]

\node[circle, fill, minimum size=3pt, inner sep=0pt, label=left:$\gf_j$] (gfj) {};
\node[circle, fill, minimum size=3pt, inner sep=0pt, label=right:$\gt_j$] (gtj) [right=1.2cm of gfj] {};

\node[circle, fill, minimum size=3pt, inner sep=0pt, label=below left:$\gf_i$] (gf)   [below=1.3cm of gfj] {};
\node[circle, fill, minimum size=3pt, inner sep=0pt, label=below right:$\gt_i$] (gt)   [below=1.3cm of gtj] {};

\node[color=blue, circle, draw, minimum size=3.5mm, inner sep=0pt, label={[text=gray]below:$m$}] (ggf)  [left=1cm of gf] {};

\node[color=blue, circle, draw, minimum size=3.5mm, inner sep=0pt, label={[text=gray]below:$m$}] (ggt) [right=1cm of gt] {};

\node[circle, fill, minimum size=3pt, inner sep=0pt, label=below left:$\lg_i$] (lg)   [below=1cm of $(gf)!0.5!(gt)$] {};

\node[color=blue, circle, draw, minimum size=3.5mm, inner sep=0pt, label={[text=gray]right:$m$}] (lgg)  [below=1cm of lg] {};

\draw[red] (gt) -- (gf);
\draw[red] (gt) -- (gtj);
\draw[red] (gf) -- (gfj);
\draw [gray, -, shorten >=0cm, shorten <=0.3cm] (gt) -- (gfj);
\draw [gray, ->, shorten >=1.5cm, shorten <=0cm] (gt) -- (gfj);
\draw [gray](gt) -- (lg);
\draw [gray, -, shorten >=0cm, shorten <=0.3cm] (gf) -- (gtj);
\draw [gray, ->, shorten >=1.5cm, shorten <=0cm] (gf) -- (gtj);
\draw [gray](gf) -- (lg);
\draw[blue, thick] (gt) -- (ggt);
\draw[blue, thick] (gf) -- (ggf);
\draw[blue, thick] (lg) -- (lgg);

\end{tikzpicture}
    \caption{Not-gadget of gate $\GG_i=\neg\GG_j$.}
    \label{ASHG_fig:Not-gadget}
\end{minipage}
\hfill
\begin{minipage}{0.45\textwidth}
    \centering
    \begin{tikzpicture}[
    every node/.style={inner sep=2pt},
    >=stealth
]

\node[circle, fill, minimum size=3pt, inner sep=0pt, label=above:$\gf_j$] (gfj) {};
\node[circle, fill, minimum size=3pt, inner sep=0pt, label=above:$\gf_k$] (gfk) [right=0.8cm of gfj] {};
\node[circle, fill, minimum size=3pt, inner sep=0pt, label=above:$\gt_j$] (gtj) [right=1.2cm of gfk] {};
\node[circle, fill, minimum size=3pt, inner sep=0pt, label=above:$\gt_k$] (gtk) [right=0.8cm of gtj] {};

\node[circle, fill, minimum size=3pt, inner sep=0pt, label=below:$\gf_i$] (gf)   [below=2.5cm of $(gfj)!0.5!(gfk)$] {};
\node[circle, fill, minimum size=3pt, inner sep=0pt, label=above right:$\gt_i$] (gt)   [below=2.5cm of $(gtj)!0.5!(gtk)$] {};

\node[color=blue, circle, draw, minimum size=3.5mm, inner sep=0pt, label={[text=gray]below:$m$}] (ggf) [below left=0.4cm and 0.8cm of gf] {};

\node[color=blue, circle, draw, minimum size=3.5mm, inner sep=0pt, label={[text=gray]below:$m$}] (ggt)  [below right=0.4cm and 0.8cm of gt] {};

\node[circle, fill, minimum size=3pt, inner sep=0pt, label=above:$\lg_i$] (lg)   [below=0.6cm of $(gf)!0.5!(gt)$] {};

\node[color=blue, circle, draw, minimum size=3.5mm, inner sep=0pt, label={[text=gray]right:$m$}] (lgg) [below=0.7cm of lg] {};

\node[circle, fill, minimum size=3pt, inner sep=0pt, label=above left:$\wand_i$] (w)   [left=1.2cm of $(gf)!0.65!(gfj)$] {};
\node[circle, fill, minimum size=3pt, inner sep=0pt, label=above left:$\zand_i$] (z)  [below=1cm of w] {};

\node[circle, fill, minimum size=3pt, inner sep=0pt, label=above:$\lw_i$] (lw)  [above=0.8cm of w] {};
\node[circle, fill, minimum size=3pt, inner sep=0pt, label=above left:$\lz_i$] (lz)  [below=0.8cm of z] {};

\node[color=blue, circle, draw, minimum size=3.5mm, inner sep=0pt, label={[text=gray]below:$m$}] (lww) [left=0.7cm of lw] {};

\node[color=blue, circle, draw, minimum size=3.5mm, inner sep=0pt, label={[text=gray]below:$m$}] (lzz)  [left=0.7cm of lz] {};

\node[color=blue, circle, draw, minimum size=3.5mm, inner sep=0pt, label={[text=gray]below:$m$}] (ww) [left=0.7cm of w] {};

\node[color=blue, circle, draw, minimum size=3.5mm, inner sep=0pt, label={[text=gray]below:$m$}] (zz) [left=0.7cm of z] {};

\draw[red] (gt) -- (gf);
\draw[red] (gt) -- (gfj);
\draw[red] (gt) -- (gfk);
\draw[red] (w) -- (z);
\draw[red] (w) -- (gfj);
\draw[red] (w) -- (gtk);
\draw[red] (w) -- (gt);
\draw[red] (z) -- (gtj);
\draw[red] (z) -- (gfk);
\draw[red] (z) -- (gt);
\draw[blue, thick] (gt) -- (ggt);
\draw[blue, thick] (gf) -- (ggf);
\draw[blue, thick] (w) -- (ww);
\draw[blue, thick] (z) -- (zz);
\draw[blue, thick] (lg) -- (lgg);
\draw[blue, thick] (lw) -- (lww);
\draw[blue, thick] (lz) -- (lzz);

\draw [gray, -, shorten >=0cm, shorten <=0.38cm] (w) -- (lw);
\draw [gray, ->, shorten >=0.4cm, shorten <=0cm] (w) -- (lw);
\draw [gray, -, shorten >=0cm, shorten <=0.38cm] (z) -- (lz);
\draw [gray, ->, shorten >=0.4cm, shorten <=0cm] (z) -- (lz);

\draw [gray, -, shorten >=0cm, shorten <=0.5cm] (gt) -- (gtj);
\draw [gray, ->, shorten >=1.9cm, shorten <=0cm] (gt) -- (gtj);

\draw [gray, -, shorten >=0cm, shorten <=0.5cm] (gt) -- (gtk);
\draw [gray, ->, shorten >=1.9cm, shorten <=0cm] (gt) -- (gtk);

\draw [gray, -, shorten >=0cm, shorten <=0.5cm] (gf) -- (gfj);
\draw [gray, ->, shorten >=1.9cm, shorten <=0cm] (gf) -- (gfj);

\draw [black, -, shorten >=0cm, shorten <=0.5cm] (gf) -- (gfk);
\draw [black, ->, shorten >=1.9cm, shorten <=0cm] (gf) -- (gfk);

\draw [gray, ->, shorten >=0.59cm, shorten <=0cm] (lg) -- (gf);
\draw [green, ->, shorten >=0.59cm, shorten <=0cm] (gf) -- (lg);

\draw [gray, ->, shorten >=0.59cm, shorten <=0cm] (lg) -- (gt);
\draw [black, ->, shorten >=0.59cm, shorten <=0cm] (gt) -- (lg);

\draw [gray](gf) -- (w);

\draw [black, ->, shorten >=0.77cm, shorten <=0cm] (gf) -- (z);
\draw [gray, ->, shorten >=0.77cm, shorten <=0cm] (z) -- (gf);

\end{tikzpicture}
    \caption{And-gadget of gate $\GG_i=\GG_j\land \GG_k$.}
    \label{ASHG_fig:And-gadget}
\end{minipage}

\caption{The main gadgets of the proof of \Cref{thm:ASHG_general}, depicted along with their input agents. 
Black dots represent single agents. 
Blue circles represent sets of one-way replicas, with a blue edge to their origin;
valuations concerning them can be derived from their origin.
Non-blue edge colors indicate valuation: green = $3$, black = $2$, gray = $1$, and red = $-\infty$.
Undirected edges imply mutual valuation of the respective color.
We omit the $-\infty$ edges connecting the inputs of a gate to each other, or involving $L$-agents.
Apart from those, omitted edges indicate valuation of $0$, even if a directed edge exists in the opposite direction (e.g., $\mathbf{v}_{\gt_j}(\gf_i)=0$ in the Not-gadget).}
\label{fig:ASHG_gadgets}
\end{figure}
\begin{figure}[!htb]
    \centering

\begin{tikzpicture}[node distance=1.8cm]

\node[ellipse, draw=black, fill=white, minimum width=7cm, minimum height=2cm, align=center] (layer1) 
{$X$-agents (where $|X|>N/2$)};

\node[circle, fill=black, inner sep=1pt, label=left:$x_1$] (x1) at ($(layer1.south)+(-2,0.4)$) {};
\node[circle, fill=black, inner sep=1pt, label=left:$x_2$] (x2) at ($(layer1.south)+(0,0.4)$) {};
\node[circle, fill=black, inner sep=1pt, label=left:$x_3$] (x3) at ($(layer1.south)+(2,0.4)$) {};

\node[circle, draw=black, fill=white, minimum size=1cm] (agd1) [below=1cm of layer1, xshift=-3cm] {$\ADD_1$};
\node[circle, draw=black, fill=white, minimum size=1cm] (agd2) [below=1cm of layer1] {$\ADD_2$};
\node[circle, draw=black, fill=white, minimum size=1cm] (agd3) [below=1cm of layer1, xshift=3cm] {$\ADD_3$};

\draw[gray, thick] (agd1.north) -- (x1);
\draw[gray, thick] (agd2.north) -- (x2);
\draw[gray, thick] (agd3.north) -- (x3);

\draw[dashed, thick, gray] (0,-5) -- (0,-13) coordinate[at start] (vline);

\node[circle, draw=black, fill=white, minimum size=1cm] (gd31) [left=0.9cm of vline, yshift=-0.5cm] {$\DD_3^1$};
\node[circle, draw=black, fill=white, minimum size=1cm] (gd21) [left=1cm of gd31] {$\DD_2^1$};
\node[circle, draw=black, fill=white, minimum size=1cm] (gd11) [left=1cm of gd21] {$\DD_1^1$};
\node[circle, draw=black, fill=white, minimum size=1cm] (gd12) [right=0.9cm of vline, yshift=-0.5cm] {$\DD_1^2$};
\node[circle, draw=black, fill=white, minimum size=1cm] (gd22) [right=1cm of gd12] {$\DD_2^2$};
\node[circle, draw=black, fill=white, minimum size=1cm] (gd32) [right=1cm of gd22] {$\DD_3^2$};

\foreach \gadget in {gd11,gd12} {
\draw[gray, thick] (\gadget.north) -- (agd1.south);
}
\foreach \gadget in {gd21,gd22} {
\draw[gray, thick] (\gadget.north) -- (agd2.south);
}
\foreach \gadget in {gd31,gd32} {
\draw[gray, thick] (\gadget.north) -- (agd3.south);
}

\foreach \angle in {36,72,108,144} {
  \foreach \gadget in {gd11,gd21,gd31,gd12,gd22,gd32} {
    \draw[gray, thick] (\gadget.south) -- ++(\angle:-0.8);
  }
}

\def\LayerFourText{And- and Not-gates gadgets}
\node[draw=black, fill=white, text=gray, minimum width=6.4cm, minimum height=1.3cm] (labelRight) at (-3.7,-7.7) {\LayerFourText};
\node[draw=black, fill=white, text=gray, minimum width=6.4cm, minimum height=1.3cm] (labelLeft) at (3.7,-7.7) {\LayerFourText};

\node[circle, draw=black, fill=white, minimum size=1cm, minimum height=1.5cm] (out3L) [left=0.4cm of vline, yshift=-5cm] {};
\node[circle, fill=black, inner sep=1pt, label=above:$\of_3$] (o30L) at ($(out3L.center)+(-0.5,-0.3)$) {};
\node[circle, fill=black, inner sep=1pt, label=above:$\ot_3$] (o31L) at ($(out3L.center)+(0.5,-0.3)$) {};

\node[circle, draw=black, fill=white, minimum size=1cm, minimum height=1.5cm] (out2L) [left=1cm of out3L] {};
\node[circle, fill=black, inner sep=1pt, label=above:$\of_2$] (o20L) at ($(out2L.center)+(-0.5,-0.3)$) {};
\node[circle, fill=black, inner sep=1pt, label=above:$\ot_2$] (o21L) at ($(out2L.center)+(0.5,-0.3)$) {};

\node[circle, draw=black, fill=white, minimum size=1cm, minimum height=1.5cm] (out1L) [left=1cm of out2L] {};
\node[circle, fill=black, inner sep=1pt, label=above:$\of_1$] (o10L) at ($(out1L.center)+(-0.4,-0.3)$) {};
\node[circle, fill=black, inner sep=1pt, label=above:$\ot_1$] (o11L) at ($(out1L.center)+(0.4,-0.3)$) {};

\node[circle, draw=black, fill=white, minimum size=1cm, minimum height=1.5cm] (out1R) [right=0.4cm of vline, yshift=-5cm] {};
\node[circle, fill=black, inner sep=1pt, label=above:$\of_1$] (o10R) at ($(out1R.center)+(-0.5,-0.3)$) {};
\node[circle, fill=black, inner sep=1pt, label=above:$\ot_1$] (o11R) at ($(out1R.center)+(0.5,-0.3)$) {};

\node[circle, draw=black, fill=white, minimum size=1cm, minimum height=1.5cm] (out2R) [right=1cm of out1R] {};
\node[circle, fill=black, inner sep=1pt, label=above:$\of_2$] (o20R) at ($(out2R.center)+(-0.5,-0.3)$) {};
\node[circle, fill=black, inner sep=1pt, label=above:$\ot_2$] (o21R) at ($(out2R.center)+(0.5,-0.3)$) {};

\node[circle, draw=black, fill=white, minimum size=1cm, minimum height=1.5cm] (out3R) [right=1cm of out2R] {};
\node[circle, fill=black, inner sep=1pt, label=above:$\of_3$] (o30R) at ($(out3R.center)+(-0.5,-0.3)$) {};
\node[circle, fill=black, inner sep=1pt, label=above:$\ot_3$] (o31R) at ($(out3R.center)+(0.5,-0.3)$) {};

\foreach \gadget in {out1L,out2R,out3R} {
   \foreach \angle in {60,120} {
    \draw[gray, thick] (\gadget.north) -- ++(\angle:0.8);
  }
}
\foreach \gadget in {out2L,out3L,out1R} {
    \draw[gray, thick] (\gadget.north) -- ++(90:0.8);
  }

\node[circle, fill=black, inner sep=1.5pt, label=left:$v_1$] (v1) [below=2cm of out2L] {};
\node[circle, fill=black, inner sep=1.5pt, label=right:$v_2$] (v2) [below=2cm of out2R] {};

\def\suffixL{1L}
\foreach \i in {1,2,3} {
    \def\out{o\i\suffixL}
    \draw[-, shorten >=0cm, shorten <=0.5cm] (v1) -- (\out);
    \draw[->] (v1) -- ($(v1)!1cm!(\out)$) node[at end, right] {$2^\i$};
}

\def\suffixR{1R}
\foreach \i in {1,2,3} {
    \def\out{o\i\suffixR}
    \draw[-, shorten >=0cm, shorten <=0.5cm] (v2) -- (\out);
    \draw[->] (v2) -- ($(v2)!1cm!(\out)$) node[at end, right] {$2^\i$};
}

\node[anchor=west, gray, align=center] at ($(vline |- agd2)+(-9,0)$) {Assignment\\gadgets};
\node[anchor=west, gray, align=center] at ($(vline |- gd22)+(-9,0)$) {Copy-\\gadgets};
\node[anchor=west, gray, align=center] at ($(vline |- out2R)+(-9,0)$) {Output \\gadgets};
\node[anchor=west, gray, align=center] at ($(vline |- v2)+(-9,0)$) {Voter agents};
\end{tikzpicture}

\caption{Structural schematic of the reductions used in the proofs of \Cref{thm:ASHG_general,thm:FHG}, with $n=3$ and $m=2$ (for the illustration we ignore the assumption that $m\ge 3$). 
Gray edges indicate interactions between gadgets (rather than individual edges within the game). 
The vertical dashed line separates $C_1$ and $C_2$.
Only the core agents of the output gates are shown, illustrating their connections to the voters.
Black edges represent the voters' valuations, while omitted $(v_i, \of_j)$ edges correspond to valuations of $0$.
Some output gates are arbitrarily depicted as And-gates and others as Not-gates, as can be inferred from the number of incoming gates.}
\label{fig:Overview}
\end{figure}

\subsection{Correctness of the Reduction}
\label{sec_ASHG:correctness}
We proceed to show that the $\CKTCondorcet{}$ instance $\fCC$ admits a Condorcet string if and only if the reduced ASHG instance $\langle N,\mathbf{v}\rangle$ described above admits a strongly popular partition.

We begin with the introduction of a useful definition, followed by a lemma concerning it.
 
\begin{definition}
\textbf{(ASHG Natural Partition).}
For any given bit string $\XX\in \bits{n}$, we let $\pi_{\XX}$ denote the following partition, also referred to as the \emph{natural partition} corresponding to $\XX$.
\begin{itemize}
    \item All agents in $X\cup V$ are in a coalition denoted $\MC_{\XX}$ (for denoting a \emph{main coalition} to which we will add further agents).
    \item For every input bit $\XX_i$, if $\XX_i=1$ then $\Ta_i\subseteq \MC_{\XX}$ and $\Fa_i\cup\La_i\in \pi_{\XX}$, and if $\XX_i=0$ then $\Fa_i\subseteq \MC_{\XX}$ and $\Ta_i\cup\La_i\in \pi_{\XX}$.
    \item For every gate $\GG_i$, if $\XX(\GG_i)=1$ then $\Tg_i\subseteq \MC_{\XX}$ and $\Fg_i\cup\Lg_i\in\pi_{\XX}$, and if $\XX(\GG_i)=0$ then $\Fg_i\subseteq\MC_{\XX}$ and $\Tg_i\cup\Lg_i\in\pi_{\XX}$.
    \item For every And-gate $\GG_i=\GG_j\land \GG_k$:
    \begin{itemize}
        \item if $\XX(\GG_j)=\XX(\GG_k)=1$ or $\XX(\GG_j)=\XX(\GG_k)=0$, then $\{\Wand_i\cup\Lw_i,\Zand_i\cup\Lz_i\}\subseteq\pi_{\XX}$.
        \item if $\XX(\GG_j)=0$ and $\XX(\GG_k)=1$ then $\{\Wand_i\cup\Lw_i,\Lz_i\}\subseteq\pi_{\XX}$, and $\Zand_i\subseteq\MC_{\XX}$.
        \item if $\XX(\GG_j)=1$ and $\XX(\GG_k)=0$ then $\{\Zand_i\cup\Lz_i,\Lw_i\}\subseteq\pi_{\XX}$, and $\Wand_i\subseteq\MC_{\XX}$.
    \end{itemize}
\end{itemize}
\end{definition}

Our main goal is to show the following claim.
\begin{lemma}
\label{lem_ASHG:natural_partitions_beat_others}
Let $\XX\in \bits{n}$ be an arbitrary string, and let $\pi$ be a Pareto-optimal partition such that $\phi(\pi_{\XX},\pi)\le 0$. Then $\pi=\pi_{\XX'}$ for some string $\XX'$.
\end{lemma}

Proving \Cref{lem_ASHG:natural_partitions_beat_others} requires several auxiliary lemmas. Once it is proved, the proof of \Cref{thm:ASHG_general} will follow without much effort.
Let $\XX\in\bits{n}$ and let $\pi$ be a Pareto-optimal partition with $\phi(\pi_{\XX},\pi)\le 0$. 
Denote $\MC=\pi(\xx_1)$ (we will see that $\MC$ will be the ``main'' coalition in $\pi$, similarly to $\MC_{\XX}$ in $\pi_{\XX}$).
We establish several separate claims concerning the structure of $\pi$, and the popularity margins of different agent sets in the vote of $\pi$ against $\pi_{\XX}$.

\begin{lemma}
\label{lem_ASHG:L_separated}
Let $\ell,\ell'\in L$, where neither is a replica of the other (they may be replicas of other $L$-agents). Then $\ell'\notin\pi(\ell)$.
\end{lemma}

\begin{proof}
Assume for contradiction that $\ell'\in\pi(\ell)$.
Notice that, since neither is a replica of the other, every agent in $N$ assigns $-\infty$ to at least one of $\ell$ and $\ell'$, and therefore every agent in $\pi(\ell)$ obtains negative utility. Hence, it is a Pareto improvement to dissolve this coalition into singletons, a contradiction to the Pareto optimality of $\pi$.
\end{proof}

\begin{lemma}
\label{lem_ASHG:X_margin}
Let $x\in X$. We have that $\phi_x(\pi_{\XX},\pi)\geq 0$.
\end{lemma}

\begin{proof}
Agent $x$ only assigns a positive value to other $(X\cup V)$-agents, and only assigns a negative value to $L$-agents. 
Recall that by definition of $\pi_{\XX}$ we have that $X\subseteq \MC_{\XX}$ and $L\cap \MC_{\XX}=\emptyset$. 
Thus, $x$ obtains the maximum possible utility in $\pi_{\XX}$. 
Therefore we have $\phi_x(\pi_{\XX},\pi)\geq 0$.
\end{proof}

\begin{lemma} 
\label{lem_ASHG:X_V_coalition}
It holds that $X\cup V\subseteq \MC$.
\end{lemma}

\begin{proof}
Assume for contradiction that there exists some $\ga\in (X\cup V)\bs\{x_1\}$ such that $\ga\notin\pi(x_1)=\MC$. 
Then all $X$-agents prefer $\pi_{\XX}$ over $\pi$, and so $\phi_{X}(\pi_{\XX},\pi) = |X|$. 
By design, $|X|>\frac{|N|}{2}$, and therefore $\pi_{\XX}$ is more popular than $\pi$, a contradiction. 
\end{proof}

\begin{lemma}
\label{lem_ASHG:L_notin_MC}
It holds that $L\cap \MC=\emptyset$.
\end{lemma}

\begin{proof} 
Assume otherwise. Then by \Cref{lem_ASHG:X_V_coalition} we have that all $X$-agents obtain a negative utility in $\pi$, and thus prefer $\pi_{\XX}$. 
By design, $|X|>\frac{|N|}{2}$, and therefore $\pi_{\XX}$ is more popular than $\pi$, a contradiction.
\end{proof}

\begin{lemma}
\label{lem_ASHG:assignment_margin}
Let $\XX_i$ be an input bit, and denote the corresponding assignment gadget by $\ADD_i=\Fa_i\cup\Ta_i\cup\La_i$. Then we have that $\phi_{\ADD_i}(\pi_{\XX},\pi)\geq 0$.
\end{lemma}

\begin{proof}
Assume for contradiction that $\phi_{\ADD_i}(\pi_{\XX},\pi)<0$. 
Note that by \Cref{lem_ASHG:X_V_coalition,lem_ASHG:L_notin_MC} we have that $\ls_i\notin\pi(x_i)$, and therefore agents in $\Ta_i\cup\Fa_i$ cannot prefer $\pi$ over $\pi_{\XX}$. 
Hence, some agent $\ell\in\La_i$ must prefer $\pi$, and so we must have $\{\asst_i,\assf_i\}\subseteq \pi(\ell)$. 
This implies that all agents in $\Ta_i\cup\Fa_i$ prefer $\pi_{\XX}$, a contradiction to $\phi_{\ADD_i}(\pi_{\XX},\pi)<0$.
\end{proof}

\begin{lemma}
\label{lem_ASHG:w_z_balance}
Let $\GDD_i$ be an And-gadget. 
If an agent in $\Wand_i$ prefers $\pi$ over $\pi_{\XX}$, then $\phi_{\Fg_i\cup\Wand_i\cup\Lw_i}(\pi_{\XX},\pi)\ge m+1$ and $\phi_{\Zand_i\cup\Lz_i}(\pi_{\XX},\pi)\ge 0$.
If an agent in $\Zand_i$ prefers $\pi$ over $\pi_{\XX}$, then $\phi_{\Fg_i\cup\Zand_i\cup\Lz_i}(\pi_{\XX},\pi)\ge m+1$ and $\phi_{\Wand_i\cup\Lw_i}(\pi_{\XX},\pi)\ge 0$.
\end{lemma}

\begin{proof}
First, observe that agents in $\Lw\cup\Lz$ obtain maximal utility in $\pi_{\XX}$ and thus cannot prefer $\pi$.
Hence, if an agent $\ga\in\Wand_i$ prefers $\pi$, then $\{\gf_i,\lw_i\}\subseteq\pi(\ga)$, implying that $\phi_{\Zand_i}(\pi_{\XX},\pi)\ge 0$ and also $\phi_{\Fg_i\cup\Lw_i}(\pi_{\XX},\pi)=2m+2$, which yields $\phi_{\Fg_i\cup\Wand_i\cup\Lw_i}(\pi_{\XX},\pi)\ge m+1$.

If an agent of $\Zand_i$ prefers $\pi$, the argument is symmetric.
\end{proof}

\begin{lemma}
\label{lem_ASHG:w_z_gf_tradeoff}
Let $\GG_i=\GG_j\land\GG_k$ be an And-gate. If $\Fg_i\cap\pi(\lg_i)=\emptyset$ or $\Tg_i\cap\pi(\lg_i)\neq\emptyset$, then $\phi_{\Fg_i\cup\Wand_i\cup\Zand_i\cup\Lw_i\cup\Lz_i}(\pi_{\XX},\pi)\geq 0$.
\end{lemma}

\begin{proof}
If any agent in $\Wand_i\cup\Zand_i$ prefers $\pi$ over $\pi_{\XX}$, then we are done by \Cref{lem_ASHG:w_z_balance}.
Assume, therefore, that agents of $\Wand_i\cup\Zand_i$ do not prefer $\pi$.
Observe further that no agent in $\Lw_i\cup\Lz_i$ may prefer $\pi$, as they obtain maximal utility in $\pi_{\XX}$.
Thus, if no agent in $\Fg_i$ prefers $\pi$, we are again done.
So assume that some agent $\ga\in\Fg_i$ prefers $\pi$. 
Then we must have $\lg_i\notin\pi(\ga)$, since if $\lg_i\in\pi(\ga)$ then by the statement of this lemma we have $\Tg_i\cap\pi(\lg_i)\neq\emptyset$, and so $\ga$ prefers $\pi_{\XX}$, a contradiction. 
Hence, in order for $p$ to prefer $\pi$ we must have $\{\zand_i,\gf_k\}\subseteq\pi(\ga)$ or $|\pi(\ga)\cap\{\wand_i,\zand_i,\gf_k,\gf_j\}|\ge 3$. Thus, either all agents in $\Wand_i$ or all agents in $\Zand_i$ prefer $\pi_{\XX}$. 
Since we established that the only agents in $\Fg_i\cup\Wand_i\cup\Zand_i\cup\Lw_i\cup\Lz_i$ who may prefer $\pi$ are those of $\Fg_i$, the above suffices to imply that
$\phi_{\Fg_i\cup\Wand_i\cup\Zand_i\cup\Lw_i\cup\Lz_i}(\pi_{\XX},\pi)\geq 0$.
\end{proof}

\begin{lemma}
\label{lem_ASHG:gate_margins}
Let $\GG_i$ be a gate with corresponding gadget $\GDD_i$. Then we have $\phi_{\GDD_i}(\pi_{\XX},\pi)\geq 0$.
\end{lemma}

\begin{proof}
Assume for contradiction that $\phi_{\GDD_i}(\pi_{\XX},\pi)<0$.
We make a case distinction based on the type of gate.

\emph{Case 1:} Suppose $\GG_i=\XX_i$ is a Copy-gate. 
If some $\ga\in\Tg_i$ 
prefers $\pi$ over $\pi_{\XX}$ then $\{\asst_i,\lg_i\}\subseteq\pi(\ga)$.
Therefore, $u_{\lg_i}<0$, while any $\ell\in r(\lg_i)$ is either separated from $\lg_i$ or in the same coalition as $\asst_i$ under $\pi$.
Thus, all agents in $\Lg_i$ prefer $\pi_{\XX}$ over $\pi$.
Moreover, $\asst_i\in\pi(\lg_i)$ implies that all agents in $\Fg_i$ obtain at most the utility they obtain in $\pi_{\XX}$ (since being grouped with $\lg_i$ cannot benefit them).
Thus, $\phi_{\GDD_i}(\pi_{\XX},\pi)\geq 0$, a contradiction.
If some $\ga\in\Fg_i$ prefers $\pi$ over $\pi_{\XX}$ we reach a similar contradiction.
Hence, some $\ell\in\Lg_i$ must prefer $\pi$, and so we must have $\{\gt_i,\gf_i\}\subseteq \pi(\ell)$. This implies that all agents in $\Tg_i\cup\Fg_i$ prefer $\pi_{\XX}$, which contradicts $\phi_{\GDD_i}(\pi_{\XX},\pi)<0$.

\emph{Case 2:} If $\GG_i=\lnot \GG_j$ is a Not-gate, the proof is analogous to Case 1. 

\emph{Case 3:} Suppose $\GG_i=\GG_j\land\GG_k$ is an And-gate. 
If some $\ell\in\Lg_i$ prefers $\pi$ over $\pi_{\XX}$ then $\{\gt_i,\gf_i\}\subseteq\pi(\ell)$, implying that all agents in $\Tg_i\cup\Fg_i$ prefer $\pi_{\XX}$. 
Since $\gf_i\in\pi(\ell)$, it follows that no agent in $\Wand_i\cup\Zand_i$ may prefer $\pi$, a contradiction to $\phi_{\GDD_i}(\pi_{\XX},\pi)<0$.
If some $\ga\in\Tg_i$ prefers $\pi$ over $\pi_{\XX}$ then $\lg_i\in\pi(\ga)$ and also $|\{\gt_j,\gt_k\}\cap\pi(\ga)|\geq 1$. 
This implies that any $\ell\in\Lg_i$ prefers $\pi_{\XX}$.
Moreover, since $\lg_i\in\pi(\ga)$ we may apply \Cref{lem_ASHG:w_z_gf_tradeoff}, and we reach a contradiction to $\phi_{\GDD_i}(\pi_{\XX},\pi)<0$.
So agents in $\Lg_i\cup\Tg_i$ cannot prefer $\pi$, and thus we are left only with the possibility that some $\ga\in\Fg_i\cup\Wand_i\cup\Zand_i$ prefers $\pi$ (since $\Lw_i\cup\Lz_i$-agents obtain maximal utility in $\pi_{\XX}$). 
However, if an agent in $\Wand_i\cup\Zand_i$ prefers $\pi$ we reach a contradiction by \Cref{lem_ASHG:w_z_balance}.

Hence, $\ga\in\Fg_i$, and so all of $\Fg_i$ are in the same coalition. Furthermore, if 
$\Fg_i\cap\pi(\lg_i)=\emptyset$ then \Cref{lem_ASHG:w_z_gf_tradeoff} implies that $\phi_{\Fg_i\cup\Wand_i\cup\Zand_i\cup\Lw_i\cup\Lz_i}(\pi_{\XX},\pi)\geq 0$ while the previous cases of this lemma establish that $\phi_{\Lg_i\cup\Tg_i}(\pi_{\XX},\pi)\ge 0$, a contradiction to $\phi_{\GDD_i}(\pi_{\XX},\pi)<0$. 
Otherwise, we have $\lg_i\in\pi(\ga)$ (since all of $\Fg_i$ are in the same coalition). Moreover, since $\ga$ prefers $\pi$, there must be another agent in $\pi(\ga)$ who $\ga$ assigns positive value to. This implies that all agents in $\Lg_i$ prefer $\pi_{\XX}$, and we again conclude that $\phi_{\GDD_i}(\pi_{\XX},\pi)\geq 0$ even if all agents of $\Fg_i$ prefer $\pi$, a contradiction.
\end{proof}

\begin{lemma}
\label{lem_ASHG:gates_dont_comply_margin}
Let $\GG_i$ be a gate with corresponding $G$-gadget $\GDD_i$. Suppose the inputs (or input) of $\GG_i$ have $\pi$-valid corresponding gadgets.
Then if $\GDD_i$ does not $\pi$-comply with its inputs, we have that $\phi_{\GDD_i}(\pi_{\XX},\pi)\geq m+1$.
\end{lemma}
\begin{proof} 
First, notice that any agent $\ell\in\Lg_i$ can only prefer $\pi$ if $\{\gt_i,\gf_i\}\subseteq\pi(\ell)$, but then all agents in $\Tg_i\cup\Fg_i$ prefer $\pi_{\XX}$.
If $\GDD_i$ is not an And-gadget then the above suffices to yield the claim.
If $\GDD_i$ is an And-gadget, we additionally observe that since this would imply that $\gf_i\in\pi(\ell)$, all agents in $\Wand_i\cup\Zand_i$ cannot prefer $\pi$, and the result follows. 

So we may assume that agents in $\Lg_i$ do not prefer $\pi$.
Furthermore, since the input gadgets of $\GG_i$ are $\pi$-valid, \Cref{lem_ASHG:L_notin_MC,lem_ASHG:L_separated} imply that $\pi(\lg_i)$ contains no core agents of the input gadgets; hence no agent of $\Tg_i$ may prefer $\pi$.
We now make a case distinction based on the gadget type of $\GDD_i$.

\emph{Case 1:} Suppose $\GDD_i$ is a Copy-gadget, with input $A$-gadget $\ADD_i$. 
Since $\ADD_i$ is $\pi$-valid, assume without loss of generality that $\asst_i\in\MC$ and $\assf_i\in\pi(\ls_i)$ (the other case is symmetric). 
Then, by assumption that $\GDD_i$ does not $\pi$-comply with $\ADD_i$, we have that either $\gf_i\notin\pi(\lg_i)$ or $\gt_i\notin\MC$. 
Thus, by \Cref{lem_ASHG:L_separated} we have $\pi(\lg_i)\ne\pi(\assf_i)$, hence no agent of $\Fg_i$ may prefer $\pi$. 
Now, if $\gf_i\notin\pi(\lg_i)$ then all agents of $\Fg_i$ prefer $\pi_{\XX}$ and we are done.
Otherwise all agents of $\Tg_i$ prefer $\pi_{\XX}$, and again we are done.

\emph{Case 2:} If $\GDD_i$ is a Not-gadget, the proof is analogous to Case 1.

\emph{Case 3:} Suppose $\GDD_i$ is an And-gadget with $\pi$-valid input $G$-gadgets $\GDD_j$ and $\GDD_k$, and that $\GDD_i$ does not $\pi$-comply with its inputs. 
Recall that agents in $\Lg_i\cup\Tg_i$ do not prefer $\pi$. 
If some agent of $\Wand_i\cup\Zand_i$ prefers $\pi$ over $\pi_{\XX}$, then by \Cref{lem_ASHG:w_z_balance} we are done.
Hence, no agent in $\Wand_i\cup\Zand_i$ may prefer $\pi$ over $\pi_{\XX}$, and the agents of $\Fg_i$ remain the only ones in $\GDD_i$ who may prefer $\pi$; 
call this observation $\spadesuit$.
If any two of $\{\wand_i,\zand_i,\lg_i\}$ share a coalition in $\pi$ then $\phi_{\Wand_i\cup\Zand_i\cup\Lg_i}(\pi_{\XX},\pi)\ge 2m+2$ and we are done by $\spadesuit$. 
So assume $\wand_i$, $\zand_i$, and $\lg_i$ are in three separate coalitions in $\pi$;
call this assumption $\clubsuit$.

We consider the following three sub-cases corresponding to the possible values of the input gates.

\emph{Case 3.1} Suppose $\GDD_j$ and $\GDD_k$ $\pi$-correspond to $\GG_j=1$ and $\GG_k=1$, and therefore $\gf_j\in\pi(\lg_j)$ and $\gf_k\in\pi(\lg_k)$.
Since $\GDD_i$ does not $\pi$-comply with its input, we have that either $\gf_i\notin\pi(\lg_i)$, or $\gt_i\notin\MC$.
If $\gf_i\notin\pi(\lg_i)$, then by $\clubsuit$ all of $\Fg_i$ prefer $\pi_{\XX}$ and we are done.
If $\gf_i\in\pi(\lg_i)$ and $\gt_i\notin\MC$ then $\phi_{\Tg_i}(\pi_{\XX},\pi)=m+1$,
while by $\clubsuit$ we have that $\phi_{\Fg_i}(\pi_{\XX},\pi)\ge 0$, and we are done by $\spadesuit$.

\emph{Case 3.2:} Suppose $\GDD_j$ and $\GDD_k$ correspond to $\GG_j=0$ and $\GG_k=0$.  
Since $\GDD_i$ does not $\pi$-comply with its inputs, we have that either $\gf_i\notin\MC$ or $\gt_i\notin\pi(\lg_i)$. 
First, suppose $\gf_i\notin\MC$. 
If $\lg_i\notin\pi(\gf_i)$ then by $\clubsuit$ we have that $\phi_{\Fg_i}(\pi_{\XX},\pi)=m+1$, and by $\spadesuit$ we are done.
If $\lg_i\in\pi(\gf_i)$, then $\phi_{\Tg_i}(\pi_{\XX},\pi)=m+1$, while by $\clubsuit$ no agent of $\Fg_i$ prefers $\pi$, and we are done by $\spadesuit$.

So we must have $\gf_i\in\MC$ and $\gt_i\notin\pi(\lg_i)$. 
Thus, we have $\phi_{\Tg_i}(\pi_{\XX},\pi)=m+1$. 
Additionally, if $\wand_i\in \MC$ then $\phi_{\Wand_i}(\pi_{\XX},\pi)\ge m+1$, if $\zand_i\in \MC$ then $\phi_{\Zand_i}(\pi_{\XX},\pi)\ge m+1$, and if $\{\wand_i,\zand_i\}\nsubseteq\MC$ then we have $\phi_{\Fg_i}(\pi_{\XX},\pi)\ge 0$. 
Hence, we have that $\phi_{\Fg_i\cup\Wand_i\cup\Zand_i}(\pi_{\XX},\pi)\geq 0$, and so $\phi_{\GDD_i}(\pi_{\XX},\pi)\ge m+1$.

\emph{Case 3.3:} Suppose $\GDD_j$ and $\GDD_k$ correspond to $\GG_j=0$ and $\GG_k=1$ respectively (although the gadget is not completely symmetric with respect to $j$ and $k$, the proof of the case where $\GG_j=1$ and $\GG_k=0$ is analogous). 
Since $\GDD_i$ does not $\pi$-comply with its inputs, we have that $\gf_i\notin\MC$ or $\gt_i\notin\pi(\lg_i)$. 
First, suppose $\gf_i\notin\MC$. 
Since the input gadgets are valid, we have $\{\gf_j,\gt_k\}\subseteq\MC\neq\pi(\gf_i)$, $\gf_k\in\pi(\lg_k)$, and $\gt_j\in\pi(\lg_j)$.
If $\lg_i\notin\pi(\gf_i)$, then by $\clubsuit$ we have that $\phi_{\Fg_i}(\pi_{\XX},\pi)=m+1$ and by $\spadesuit$ we are done.
If $\lg_i\in\pi(\gf_i)$, then $\phi_{\Tg_i}(\pi_{\XX},\pi)=m+1$, and by $\clubsuit$ we have $\phi_{\Fg_i}(\pi_{\XX},\pi)\ge 0$; thus, by $\spadesuit$ we are done.

Hence, we must have $\gf_i\in\MC$ and $\gt_i\notin\pi(\lg_i)$. 
Therefore, we have that $\phi_{\Tg_i\cup\Lg_i}(\pi_{\XX},\pi)=2m+2$.
By $\spadesuit$ we are done.
\end{proof}

\begin{lemma}
\label{lem_ASHG:assignment_MC_intersection}
Let $i\in [n]$. Then $|\{\asst_i,\assf_i\}\cap\pi(x_i)|=1$.
\end{lemma}

\begin{proof}
First, assume for contradiction that $|\{\asst_i,\assf_i\}\cap\pi(x_i)|=2$. Then it can be seen that all agents in $\ADD_i$ prefer $\pi_{\XX}$, i.e., $\phi_{\ADD_i}(\pi_{\XX},\pi)=3m+3$. Hence, by \Cref{lem_ASHG:X_margin,lem_ASHG:assignment_margin,lem_ASHG:gate_margins},
we have that $\phi(\pi_{\XX},\pi)>0$ (even if all $V$-agents prefer $\pi$), a contradiction.
Hence, we have that $|\{\asst_i,\assf_i\}\cap\pi(x_i)|\le 1$. 

Second, assume for contradiction that $|\{\asst_i,\assf_i\}\cap\pi(x_i)|=0$.
If there exist $a\in\Ta_i$ and $a'\in\Fa_i$ such that $a\in\pi(a')$, it can be seen that all agents in $\Ta_i\cup\Fa_i$ prefer $\pi_{\XX}$ over $\pi$.
Otherwise, no agent in $\ADD_i$ can prefer $\pi$, while all of $\Ta_i$ or all of $\Fa_i$ must strictly prefer $\pi_{\XX}$ (since only one of those sets can have a non-empty intersection with $\pi(\ls_i)$).
Thus, in both cases we get $\phi_{\ADD_i}(\pi_{\XX},\pi)\geq m+1$.
Hence, by \Cref{lem_ASHG:X_margin,lem_ASHG:assignment_margin,lem_ASHG:gate_margins}, 
we have that $\phi(\pi_{\XX},\pi)>0$ (even if all $V$-agents prefer $\pi$), a contradiction. Thus, $|\{\asst_i,\assf_i\}\cap\pi(x_i)|=1$.
\end{proof}

\begin{lemma}
\label{lem_ASHG:A_valid}
All assignment gadgets are $\pi$-valid.
\end{lemma}
\begin{proof}
Let $i\in [n]$. By \Cref{lem_ASHG:assignment_MC_intersection} we have $|\{\asst_i,\assf_i\}\cap\MC|=1$. Without loss of generality assume $\asst_i\in\MC$. If $\assf_i\notin \pi(\ls_i)$, then all agents in $\Fa_i\cup \La_i$ prefer $\pi_{\XX}$ over $\pi$, and agents in $\Ta_i$ cannot prefer $\pi$ over $\pi_{\XX}$. Therefore, $\phi_{\ADD_i}(\pi_{\XX},\pi)\geq 2m+2$. 
Hence, by \Cref{lem_ASHG:X_margin,lem_ASHG:assignment_margin,lem_ASHG:gate_margins},
we have that $\phi(\pi_{\XX},\pi)>0$ (even if all $V$-agents prefer $\pi$), a contradiction. Hence, $\assf_i\in \pi(\ls_i)$, and the gadget is $\pi$-valid.
\end{proof}

\textbf{Extracting an assignment from $\pi$:} At this point, we can extract from $\pi$ the input string represented by its assignment gadgets. 
For the remainder of this section, we call this string $\XX'$, where for all $i\in[n]$ $\XX'_i=1$ if and only if $\asst_i\in \MC$.

\begin{lemma}
\label{lem_ASHG:G_comply}
All gate gadgets $\pi$-comply with $\XX'$ (and are therefore $\pi$-valid).
\end{lemma}
\begin{proof}
Assume towards contradiction that some gate does not comply with $\XX'$.
Consider the gate $\GG_i$ with the smallest index that does not $\pi$-comply with $\XX'$.
Since the gates are topologically ordered, its inputs are $\pi$-valid and $\pi$-comply with $\XX'$ (if $\GG_i$ happens to be a Copy-gate, its inputs are $\pi$-valid by \Cref{lem_ASHG:A_valid}, and they $\pi$-comply with $\XX'$ by definition of $\XX'$).
Thus, by \Cref{lem_ASHG:gates_dont_comply_margin}, we have that $\phi_{\GDD_i}(\pi_{\XX},\pi)\geq m+1$.
Hence, by \Cref{lem_ASHG:X_margin,lem_ASHG:gate_margins,lem_ASHG:assignment_margin}, even if all voters prefer $\pi$ over $\pi_{\XX}$ we have $\phi(\pi_{\XX},\pi)\ge m+1-m=1$, a contradiction.
\end{proof}

\begin{lemma}
\label{lem_ASHG:non-voters_weakly_prefer_pi_x}
Let $\ga\in N\bs V$.
Then $\phi_{\ga}(\pi_{\XX},\pi)\ge 0$.
\end{lemma}
\begin{proof}
We give the analysis for non-replica agents, as the analysis of the replicas is the same as their origin.

\Cref{lem_ASHG:X_margin} settles the $X$-agents.
By \Cref{lem_ASHG:A_valid,lem_ASHG:G_comply}, $L$-agents cannot prefer $\pi$: $(\Lw\cup\Lz)$-agents already obtain maximal utility in $\pi_{\XX}$, while $(\Lg\cup\La)$-agents can never be with both core agents in a coalition.
In any gadget apart from And-gadgets, the core agents cannot prefer $\pi$ since this would require a joint coalition involving the gadget's $L$-agent and one of the core agents of the gadget's inputs, a contradiction to \Cref{lem_ASHG:A_valid,lem_ASHG:G_comply,lem_ASHG:L_separated,lem_ASHG:L_notin_MC}.
The same logic applies for $\gt_i$ of an And-gadget.

For an And-gadget $\GDD_i$ of gate $\GG_i=\GG_j\land\GG_k$, agents $\wand_i$ and $\zand_i$ cannot prefer $\pi$ as that would require $\gf_i\in\pi(\lw_i)$ or $\gf_i\in\pi(\lz_i)$, which, by \Cref{lem_ASHG:G_comply}, contradicts either \Cref{lem_ASHG:L_separated} or \Cref{lem_ASHG:L_notin_MC}.

Thus, the only remaining possibility is that agent $\gf_i$ prefers $\pi$.
If any two of $\wand_i$, $\zand_i$, and $\lg_i$ are in the same coalition of $\pi$, we obtain $\phi_{\GDD_i}(\pi_{\XX},\pi)\ge m+1$ even if all agents of $\Fg_i$ prefer $\pi$; by \Cref{lem_ASHG:X_margin,lem_ASHG:gate_margins,lem_ASHG:assignment_margin}, this implies that $\phi(\pi_{\XX},\pi)\ge 1$, a contradiction.
Hence, $\wand_i$, $\zand_i$, and $\lg_i$ are in three separate coalitions.
Thus, if $\gf_i$ prefers $\pi$ over $\pi_{\XX}$, either $\{\zand_i,\gf_k\}\subseteq\pi(\gf_i)$, or $|\{\wand_i,\zand_i,\gf_j,\gf_k\}\cap\pi(\gf_i)|\ge 3$. 
Hence, for some $t\in\{j,k\}$ all agents of $\Fg_t$ obtain negative utility in $\pi$, and thus prefer $\pi_{\XX}$.
Thus, since $\Fg_t$ are the only agents of $\GDD_t$ who can potentially prefer $\pi$ (by the previous cases covered) we have that $\phi_{\GDD_t}(\pi_{\XX},\pi)\ge m+1$.
Therefore, by \Cref{lem_ASHG:X_margin,lem_ASHG:gate_margins,lem_ASHG:assignment_margin} we have $\phi(\pi_{\XX},\pi)\ge 1$, a contradiction.
\end{proof}

\begin{lemma}
\label{lem_ASHG:replicas_with_origin}
Let $r$ be a one-way replica of an origin $\ga$.
Then $r\in\pi(\ga)$.
\end{lemma}
\begin{proof}
If not, then all $m+1$ agents in $\{\ga\}\cup r(\ga)$ must prefer $\pi_{\XX}$ over $\pi$.
Thus, by \Cref{lem_ASHG:non-voters_weakly_prefer_pi_x} we have that $\phi(\pi_{\XX},\pi)=\phi_{N\bs V}(\pi_{\XX},\pi)+\phi_{V}(\pi_{\XX},\pi)\ge m+1-m=1$, a contradiction.
\end{proof}

\begin{lemma}
\label{lem_ASHG:WZ_arrangement}
Let $\GDD_i$ be an And-gadget of gate $\GG_i=\GG_j\land\GG_k$. 
Then the following holds.
\begin{enumerate}
    \item If $\{\gt_j,\gt_k,\gt_i\}\subseteq \MC$ or $\{\gf_j,\gf_k,\gf_i\}\subseteq \MC$ then $\{\Wand_i\cup\Lw_i, \Zand_i\cup \Lz_i\}\subseteq \pi$.

    \item If $\{\gf_j,\gt_k,\gf_i\}\subseteq \MC$ then $\Zand_i\subseteq \MC$ and $\{\Wand_i\cup\Lw_i, \Lz_i\}\subseteq \pi$.

    \item If $\{\gt_j,\gf_k,\gf_i\}\subseteq \MC$ then $\Wand_i\subseteq \MC$ and $\{\Lw_i, \Zand_i\cup \Lz_i\}\subseteq \pi$.
\end{enumerate}
\end{lemma}
\begin{proof}
First, note that by \Cref{lem_ASHG:G_comply} the above enumeration lists all possible cases.
By \Cref{lem_ASHG:L_notin_MC} we have that $(\Lw_i\cup\Lz_i)\cap\MC=\emptyset$.

We begin with a useful observation.
If $\Wand_i\cap\MC=\emptyset$ (resp. $\Zand_i\cap\MC=\emptyset$) then we must have that $\lw\in \pi(\wand_i)$ (resp. $\lz\in \pi(\zand_i)$), since otherwise $\phi_{\Wand_i}(\pi_{\XX},\pi)=m+1$ (resp. $\phi_{\Zand_i}(\pi_{\XX},\pi)=m+1$), and thus by \Cref{lem_ASHG:non-voters_weakly_prefer_pi_x} $\phi(\pi_{\XX},\pi)\ge 1$, a contradiction.
In that case, by \Cref{lem_ASHG:replicas_with_origin} we also have $\Wand_i\cup\Lw_i\subseteq\pi(\lw_i)$ (resp. $\Zand_i\cup\Lz_i\subseteq\pi(\lz_i)$), and by \Cref{lem_ASHG:L_notin_MC,lem_ASHG:L_separated,lem_ASHG:G_comply} no other agent may be in that coalition apart from possibly agents of $\Wand_j\cup\Zand_j$ of some other And-gadget, but this leads to a similar contradiction since we would obtain $\phi_{\Lw_i}(\pi_{\XX},\pi)=m+1$ (resp. $\phi_{\Lz_i}(\pi_{\XX},\pi)=m+1$).
In words, if $\Wand_i$ has no representatives in $\MC$, then $\Wand_i\cup\Lw_i\in\pi$, and if $\Zand_i$ has no representatives in $\MC$, then $\Zand_i\cup\Lz_i\in\pi$.
Call this observation $\spadesuit$.
We consider each case of the lemma separately.

For the first case, assume that $\{\gt_j,\gt_k,\gt_i\}\subseteq \MC$ or $\{\gf_j,\gf_k,\gf_i\}\subseteq \MC$. 
Thus, we have $(\Wand_i\cup\Zand_i)\cap \MC= \emptyset$, as otherwise we have either $\phi_{\Wand_i}(\pi_{\XX},\pi)=m+1$ or $\phi_{\Zand_i}(\pi_{\XX},\pi)=m+1$, and by \Cref{lem_ASHG:non-voters_weakly_prefer_pi_x} we have $\phi(\pi_{\XX},\pi)\ge 1$, a contradiction.
Therefore, by $\spadesuit$ the claim holds.

Second, assume that $\{\gf_j,\gt_k,\gf_i\}\subseteq \MC$.
Then we must have $\zand_i\in \MC$ (and thus $\Zand_i\subseteq\MC$, by \Cref{lem_ASHG:replicas_with_origin}), as otherwise $\phi_{\Fg_i}(\pi_{\XX},\pi)=m+1$, yielding $\phi(\pi_{\XX},\pi)\ge 1$ by \Cref{lem_ASHG:non-voters_weakly_prefer_pi_x}.
Furthermore, we must have $\Wand_i\cap\MC=\emptyset$ as otherwise we have $\phi_{\Wand_i}(\pi_{\XX},\pi)=m+1$, reaching the same contradiction.
Hence, by $\spadesuit$ the claim holds.

Lastly, if $\{\gt_j,\gf_k,\gf_i\}\subseteq \MC$ the proof is analogous to the previous case.
\end{proof}

At this point, we can conclude that \Cref{lem_ASHG:natural_partitions_beat_others} holds, as it follows directly from \Cref{lem_ASHG:replicas_with_origin,lem_ASHG:WZ_arrangement,lem_ASHG:G_comply,lem_ASHG:A_valid,lem_ASHG:X_V_coalition}.
We can now directly prove that $\fCC$ admits a Condorcet string if and only if $\langle N,\mathbf{v}\rangle$ admits a strongly popular partition.

\textbf{Condorcet String implies Strongly Popular Partition:}
Suppose that we have a Condorcet string $\XX^*=(\XX^*_1,\dots,\XX^*_n)$ for $\fCC=\langle\CC_1,\dots,\CC_m\rangle$. 
We want to show that the corresponding natural partition $\pi^*:=\pi_{\XX^*}$ is strongly popular.
To that end, assume towards contradiction there exists a partition $\pi\neq\pi^*$ such that $\phi(\pi^*,\pi)\leq 0$.
We may assume without loss of generality that $\pi$ is Pareto-optimal\footnote{This is a well-known argument in the context of popularity \cite{BrBu20a,BG24popularity}. Briefly, if $\pi$ is at least as popular as $\pi^*$ and $\pi'$ is a Pareto improvement from $\pi$, then $\pi'$ is also at least as popular as $\pi^*$.\label{footnote:Pareto}}.
By \Cref{lem_ASHG:natural_partitions_beat_others}, we have that $\pi=\pi_{\XX'}$ for some string $\XX'\in \bits{n}$.
Thus, we have that $\XX'\ne \XX^*$, as otherwise we would have $\pi^*=\pi$, a contradiction to the choice of $\pi$.
Furthermore, all non-voter agents are indifferent between $\pi^*$ and $\pi$, as they obtain the same utility in both.

Consider a voter agent $v$ corresponding to a circuit $\CC$ with outputs $\OO_1,\dots,\OO_n$.
In both partitions, $v$ assigns $0$ to anyone in her coalition, apart from any agents $\ot_j$ of some output bit $\OO_j$ of $\CC$ who may be in her coalition, to whom she assigns $2^j$.
Denote $O=\MC_{\XX'}\cap\{\ot_j\}_{j\in[n]}$ and $O^*=\MC_{\XX^*}\cap\{\ot_j\}_{j\in[n]}$.
We thus have
\[
u_{v_j}(\pi)=\sum_{\ot_j\in O}2^j \text{ and }
u_{v_j}(\pi^*)=\sum_{\ot_j\in O^*}2^j.
\]
Hence, we have that the preference order of $v_j$ depends exclusively on $\sgn(\sum_{\ot_j\in O^*}2^j-\sum_{\ot_j\in O}2^j)$.
By definition of a natural partition, for all $j\in [n]$ we have that $\ot_j\in \MC_{\XX'}$ (and thus $\ot_j\in O$) if and only if $\XX'(\OO_j)=1$, and similarly $\ot_j\in \MC_{\XX^*}$ (and thus $\ot_j\in O^*$) if and only if $\XX^*(\OO_j)=1$.
Thus, we have that $\sum_{\ot_j\in O^*}2^j>\sum_{\ot_j\in O}2^j$ if and only if $\XX^*(\CC_j)>\XX'(\CC_j)$ (this follows from the nature of the powers of two, each of which exceeds the sum of all smaller powers appearing in the expression, mimicking the behavior of more significant bits of a bit string).
Applying this logic across all circuits, since $\XX^*$ is a Condorcet string there are more voters who prefer $\pi^*$ over $\pi$ than those who prefer the opposite. 
Hence, we have that $\phi(\pi^*,\pi)>0$, a contradiction.
Thus, $\pi^*$ is a strongly popular partition.
\qed

\textbf{Strongly Popular Partition implies Condorcet String:}
Suppose now that we have a strongly popular partition $\pi^*$ for $\langle N,\mathbf{v}\rangle$.
We wish to show that there exists a Condorcet string for $\fCC$.
A strongly popular partition is also Pareto optimal, and therefore by \Cref{lem_ASHG:natural_partitions_beat_others} there exists a bit string $\XX^*\in\bits{n}$ such that $\pi^*=\pi_{\XX^*}$ is the natural partition corresponding to $\XX^*$.
We wish to show that $\XX^*$ is a Condorcet string.
Assume towards contradiction that there exists a string $\XX^*\ne\XX'\in\bits{n}$ such that $\XX'$ beats or ties with $\XX^*$ in a majority vote among the circuits $\CC_1,\dots,\CC_m$.
Then by a similar reasoning as the converse case, there are at least as many voters who prefer $\pi_{\XX'}$ over $\pi^*$ than those who prefer the opposite, while all other agents are indifferent between the two partitions. 
Hence, $\pi^*$ is not more popular than $\pi_{\XX'}$, a contradiction to the choice of $\pi^*$.
Therefore, $\XX^*$ is a Condorcet string.
\qed


\section{Strong Popularity in Non-Negative FHGs}
\label{sec:FHG}
In this section, we prove \Cref{thm:FHG}.
In \Cref{sec_FHG:setup} we describe the reduction from \CKTCondorcet{}, and in \Cref{sec_FHG:correctness} we prove its correctness.

Throughout this section and its subsections, given an FHG we denote by $\mv{i}$ the maximum valuation assigned by agent $i$ to any other agent according to valuation function $\mathbf{v}$.
Furthermore, for a partition $\pi$, an agent $i$ and a value $t$ we write $\countv{i}{t}{\pi}:=|\{\ga\in \pi(i)\bs\{i\}\colon \mathbf{v}_{i}(\ga)=t\}|$ to denote the number of agents in $\pi(i)\bs\{i\}$ to whom $i$ assigns value $t$.
Before proceeding, we introduce a helpful definition and establish several general properties concerning strong popularity in FHGs that will be useful to prove \Cref{thm:FHG}. 

\begin{definition}
\label{def_FHG:replica-family}
\textbf{(FHG One-way replicas)}. Let $\langle N,\mathbf{v}\rangle$ be an FHG, let $\ga\in N$ and let
$S\subseteq N\bs\{\ga\}$ be a set of distinct agents.
We say that $S$ is a set of \emph{one-way replicas} (or simply \emph{replicas}) of the \emph{origin} $\ga$, if the following hold:
\begin{enumerate}
    \item For every distinct $r,r'\in S$ we have $\mathbf{v}_{r}(r')=\mathbf{v}_{r'}(r)=\frac{|S|+1}{|S|}$ (that is, replicas corresponding to the same origin assign value $\frac{|S|+1}{|S|}$ to each other).
    \item For every $r\in S$ we have $\mathbf{v}_{r}(\ga)=\frac{|S|+1}{|S|}$ and $\mathbf{v}_{\ga}(r)=1$ (that is, the origin assigns value $1$ to its replicas, while they assign value $\frac{|S|+1}{|S|}$ to the origin).
    \item For every $\ga'\in N\bs(S\cup \{\ga\})$ and every $r\in S$, we have $\mathbf{v}_{r}(\ga')=
        \begin{cases}
          0 & \text{ if } \mathbf{v}_{\ga}(\ga')=0, \\
          1 & \text{ if } \mathbf{v}_{\ga}(\ga')>0.
        \end{cases}
        $
\end{enumerate}
The set of replicas of $p$ is also denoted $r(p)=S$.
\end{definition}

Notice that if all valuations concerning an agent $\ga$ are specified, then all valuations concerning her replicas are fixed, except for $\mathbf{v}_{\ga'}(r)$ for agents $\ga'\in N\bs(S\cup \{\ga\})$, which is not specified.
This is in order to allow flexibility in the valuations that other agents assign to replicas.
Thus, when defining the FHG in \Cref{sec_FHG:setup}, we will construct replicas of some agent $\ga$, for which we would only need to specify the values $\mathbf{v}_{\ga'}(r)$ for agents $\ga'\in N\bs(S\cup \{\ga\})$.

We will use replica agents as a ``weight multiplier'' for their origin $\ga$, i.e., a tool that enables us to add more weight to the opinion of agent $\ga$ without affecting the utilities of (most) other agents.

\begin{prop}
\label{prop_FHG:valuation_threshold_of_preferences}
Let $\langle N,\mathbf{v}\rangle$ be an FHG.
Let $S\subseteq N$ be a coalition, and let $i\in S$ and $j\in N\bs S$.
Then, $\sgn(\mathbf{v}_i(j)-u_i(S))=\sgn(u_i(S\cup\{j\})-u_i(S))$ (thus, $i$ prefers adding $j$ to $S$ if and only if $\mathbf{v}_i(j)$ exceeds $i$'s utility in $S$).
\end{prop}
\begin{proof}
Let $t:=\sum_{l\in S\bs\{i\}}\mathbf{v}_i(l)$.
Then $u_i(S)=\frac{t}{|S|}$, and $u_i(S\cup \{j\})=\frac{t+\mathbf{v}_i(j)}{|S|+1}$.
Rearranging the two expressions yields the result.
\end{proof}

\begin{prop}
\label{prop_FHG:maximal_util_welcome}
Let $\langle N,\mathbf{v}\rangle$ be a non-negative FHG.
Let $S\subseteq N$ be a coalition, and let $i\in S$ with $\mv{i}>0$.
Let $\emptyset\ne S'\subseteq N\bs S$ such that for all $j\in S'$ we have $\mathbf{v}_i(j)=\mv{i}$.
Then $u_i(S\cup S')>u_i(S)$.
(That is, agents strictly benefit from adding to their coalition other agents to which they assign maximum utility).
\end{prop}
\begin{proof}
We have $u_i(S)=\frac{\sum_{l\in S\bs\{i\}}\mathbf{v}_i(l)}{|S|}\le \frac{\mv{i}(|S|-1)}{|S|}<\mv{i}$.
Hence, by \Cref{prop_FHG:valuation_threshold_of_preferences} agent $i$ prefers adding to $S$ any agent to whom $i$ assigns value $\mv{i}$.
By applying this logic iteratively, adding a single agent from $S'$ at a time, the utility of $i$ increases with each addition, concluding the proof.
\end{proof}

\subsection{Setup of the Reduction}
\label{sec_FHG:setup}
Suppose that we are given a \CKTCondorcet{} instance  $\fCC=\langle\CC_1,\dots,\CC_m\rangle$, with properties as described in \Cref{sec:prelims}.
We construct an FHG $\langle N,\mathbf{v}\rangle$ as follows. 
Let $N=X\cup A\cup G\cup L\cup W\cup Z\cup Y\cup V\cup Q$, where those disjoint sets induce different \emph{types} of agents denoted \emph{$T$-agents} for $T\in\{X,A,G,L,W,Z,Y,V,Q\}$, as detailed below. 
For $L$-agents, we make a further distinction between five sub-types, so that $L=\La\cup\Lg\cup\Lw\cup\Lz\cup\Ly$.
Note that all one-way replicas in the construction are assumed to belong to the same type as their origin agent.
Denote $K:=(n(3m+6)+1)(m+1)+\sum_{i=1}^m\big((10m+20)|\CC_i|+1\big)$. 
\begin{itemize}
    \item We create $X$-agents $\xx_1,\dots,\xx_n$, representing the input bits $\XX_1,\dots,\XX_n$ respectively. Furthermore, we create additional $X$-agents (not corresponding to any input bit) to increase the overall number of $X$-agents to $K+1$.

    \item We create $m+1$ $Q$-agents.
    
    \item For each input bit $\XX_i$ we create two \emph{assignment agents} (or $A$-agents) $\asst_i$ and $\assf_i$, and an \emph{assignment-alternative agent} (or $\La$-agent) $\ls_i$.
    Furthermore, for each of those three agents we create $m+1$ one-way replicas. 
    We denote $\Ta_i=\{\asst_i\}\cup r(\asst_i)$, $\Fa_i=\{\assf_i\}\cup r(\assf_i)$, and $\La_i=\{\ls_i\}\cup r(\ls_i)$.
    We will use the coalitions of the $\asst_i$ and $\assf_i$ agents to derive an assignment to the input string, and vice versa.
    \item For each gate $\GG_i$ (Copy-, Not-, or And-gate), we create two \emph{gate agents} (or $G$-agents) $\gt_i$ and $\gf_i$, and a \emph{gate-alternative agent} (or \emph{$\Lg$-agent}) $\lg_i$.
    Furthermore, if $\GG_i$ is a Not- or Copy-gate, then for each of $\gt_i$, $\gf_i$, and $\lg_i$ we create $m+1$ one-way replicas.
    If, however, $\GG_i$ is an And-gate, then we create $m+1$ one-way replicas for each of $\gf_i$ and $\lg_i$, while for $\gt_i$ we create $2m+3$ one-way replicas.
    We denote $\Tg_i=\{\gt_i\}\cup r(\gt_i)$, $\Fg_i=\{\gf_i\}\cup r(\gf_i)$, and $\Lg_i=\{\lg_i\}\cup r(\lg_i)$.
    We will use the coalitions of $\gt_i$ and $\gf_i$ to derive the outputs of the gates, and vice versa.
    \item For each And-gate $\GG_i$ we additionally construct a \emph{$\Wand$-agent} $\wand_i$, \emph{$\Zand$-agent} $\zand_i$, and \emph{$\Yand$-agent} $\yand_i$. Furthermore, for each of those we respectively construct an $\Lw$ agent $\lw_i$, an $\Lz$ agent $\lz_i$, and an $\Ly$ agent $\ly_i$. For each of the six agents above we construct $m+1$ one-way replicas.
    We denote $\Wand_i=\{\wand_i\}\cup r(\wand_i)$, $\Zand_i=\{\zand_i\}\cup r(\zand_i)$, $\Yand_i=\{\yand_i\}\cup r(\yand_i)$, $\Lw_i=\{\lw_i\}\cup r(\lw_i)$, $\Lz_i=\{\lz_i\}\cup r(\lz_i)$, and $\Ly_i=\{\ly_i\}\cup r(\ly_i)$.
    \item For each circuit $\CC_j$, we create a \emph{voter agent} (or \emph{$V$-agent}), denoted $v_j$. When considering a specific circuit we may omit the subscript of this agent. $v$ should not be confused with the valuation function $\mathbf{v}$.
\end{itemize}

For each input bit $\XX_i$, the set $\Ta_i\cup\Fa_i\cup\La_i$ is called the \emph{assignment gadget}, or \emph{$A$-gadget}, of $\XX_i$, denoted $\ADD_i$.
For each gate $\GG_i$, the set $\Tg_i\cup\Fg_i\cup\Lg_i$ (or, if $\GG_i$ is an And-gate, then $\Tg_i\cup\Fg_i\cup\Lg_i\cup\Wand_i\cup\Zand_i\cup \Yand_i\cup \Lw_i\cup \Lz_i\cup \Ly_i$) is called the \emph{gate gadget} of $\GG_i$, denoted $\GDD_i$.
We sometimes refer to them as Copy-, Not-, or And-gadgets when referring to a specific type of gate. 
The $A$- and $G$-agents, excluding replica agents, are referred to as \emph{core agents}. 
When considering a general (assignment or gate) gadget $\DD_i$, we let $T_i$ and $F_i$ denote the core-agent sets of the gadget (namely, $(T_i,F_i)=(\Tg_i,\Fg_i)$ if $\DD_i$ is a gate gadget, and $(T_i,F_i)=(\Ta_i,\Fa_i)$ if $\DD_i$ is an assignment gadget).
We let $R:=\{r\in r(\ga)\colon \ga\in N\}$ denote the set of all replica agents in the game.

Given a circuit $\CC_j$, we denote by $C_j$ its corresponding set of agents (which includes the gadgets of all gates of $\CC_j$, and its voter agent).
The number of non-$X$-agents is at most $K$, where we bound with the case that all gates are Copy- and And-gates.
Hence, we have that $|X|>\frac{|N|}{2}$, which will be helpful during the proof.

We now describe the valuations $\mathbf{v} = (\mathbf{v}_{\ga}\colon N\to \RR)_{\ga\in N}$ of the agents. 
Any value $\mathbf{v}_{\ga}(\ga')$ (assigned by $\ga$ to $\ga'$) not described below implies $\mathbf{v}_{\ga}(\ga')=1$ (even if $\mathbf{v}_{\ga'}(\ga)$ is specified). 
In the following description we do not include valuations of one-way replicas that can be derived from their definition using the valuations of their origin agents; we only include the exception described below \Cref{def_FHG:replica-family} concerning valuations incoming from agents external to the replica set and its origin.
\begin{itemize}
    \item Each $X$-agent assigns $v^x:=\frac{|X\cup V\cup Q|}{|X\cup V\cup Q|-1}$ to any other agent in $X\cup V\cup Q$.
    \item Each $Q$-agent assigns $v^q:=\frac{\frac{1}{2}|(G\cup A)\bs R|+1}{\frac{1}{2}|(G\cup A)\bs R|}$ to any agent in $(G\cup A)\bs R$ (namely the core agents).
    \item For each input bit $\XX_i$:
    \begin{itemize}
        \item agents $\asst_i$ and $\assf_i$ assign value $0$ to each other, and value $2$ to $\xx_i$ and to $\ls_i$;
        \item agent $\ls_i$ assigns value $2$ to agents $\asst_i$ and $\assf_i$.
    \end{itemize}
    
    \item For every gate $\GG_i$:
    \begin{itemize}
        \item agents $\gt_i$ and $\gf_i$ assign value $0$ to each other, and value $2$ to $\lg_i$;
        \item agent $\lg_i$ assigns value $2$ to agents $\gt_i$ and $\gf_i$.
     \end{itemize}
    
    \item For every Copy-gate $\GG_i=\XX_i$: 
    \begin{itemize}
        \item agent $\gt_i$ assigns $2$ to $\asst_i$, and $0$ to $\assf_i$;
        \item agent $\gf_i$ assigns $2$ to $\assf_i$, and $0$ to $\asst_i$;
        \item agent $\asst_i$ assigns $0$ to $\gf_i$;
        \item agent $\assf_i$ assigns $0$ to $\gt_i$.
    \end{itemize}
    \item For every Not-gate $\GG_i=\neg \GG_j$:
    \begin{itemize}
        \item agent $\gt_i$ assigns $2$ to $\gf_j$, and $0$ to $\gt_j$;
        \item agent $\gf_i$ assigns value $2$ to $\gt_j$, and $0$ to $\gf_j$;
        \item agent $\gt_j$ assigns $0$ to $\gt_i$;
        \item agent $\gf_j$ assigns $0$ to $\gf_i$.
    \end{itemize}
    
    \item For every And-gate $\GG_i=\GG_j\land \GG_k$:
    \begin{itemize}
        \item agent $\gt_i$ assigns $1.5$ to $\gt_j$ and to $\gt_k$, $2$ to $\lg_i$, and $0$ to any agent in $\{\gf_j,\gf_k\}\cup\Wand_i\cup\Zand_i\cup\Yand_i$.
        \item agent $\gf_i$ assigns $1.75$ to $\wand_i$ and to $\zand_i$, $1.5$ to $\yand_i$, and $1.25$ to $\gf_j$ and to $\gf_k$.
        \item agent $\wand_i$ assigns $2$ to $\lw_i$ and to $\gf_i$, and $0$ to any agent in $\{\gt_i,\gt_k,\gf_j,\zand_i,\yand_i\}$.
        \item agent $\zand_i$ assigns $2$ to $\lz_i$ and to $\gf_i$, and $0$ to any agent in $\{\gt_i,\gt_j,\gf_k,\wand_i,\yand_i\}$.
        \item agent $\yand_i$ assigns $2$ to $\ly_i$ and to $\gf_i$, and $0$ to any agent in $\{\gt_j,\gt_k,\gt_i,\wand_i,\zand_i\}$.
        \item agents $\lw_i$, $\lz_i$, and $\ly_i$ assign value $1$ to $\wand_i$, $\zand_i$, and $\yand_i$ respectively. Additionally, they each assign value $2$ to one of their own replica agents (instead of the value $1$ specified in the definition of replicas).
        \item agent $\gt_j$ assigns $0$ to any agent in $\Zand_i\cup\Yand_i$;
        \item agent $\gt_k$ assigns $0$ to any agent in $\Wand_i\cup\Yand_i$.
        \item agent $\gf_j$ assigns $0$ to any agent in $\Wand_i\cup\{\gt_i\}$;
        \item agent $\gf_k$ assigns $0$ to any agent in $\Zand_i\cup\{\gt_i\}$.
    \end{itemize}

    \item Let $v\in V$ be a voter agent. If the outputs of the circuit corresponding to $v$ are $\OO_1,...,\OO_n$ (where $\OO_n$ is the most significant bit), then for every $\OO_j,\; j\in[n]$, with corresponding $G$-agents $\ot_j$ and $\of_j$, agent $v$ assigns value $2^{j}$ to $\ot_j$, and $0$ to $\of_j$.
    Apart from that, $v$ assigns value $2^n$ to any other agent in $N\bs L$, including all replicas in $N\bs L$ (even replicas of $\ot_j$ and $\of_j$).
    
    \item Any $L$-agent assigns $0$ to all agents and vice versa, apart from the cases where other values were specified above.
    Any agent assigning $0$ to $L$ agents also assigns $0$ to their replicas.

    \item Let $r$ be a replica agent with origin $\ga\in N\bs L$.
    Let $\ga'\in N\bs(V\cup r(\ga)\cup \{\ga\})$. Then:
    \begin{itemize}
        \item if $\mathbf{v}_{\ga'}(\ga)>0$, then $\mathbf{v}_{\ga'}(r)=1$;
        \item if $\mathbf{v}_{\ga'}(\ga)=0$, then $\mathbf{v}_{\ga'}(r)=0$.
    \end{itemize}
    
\end{itemize}

Note that the $A$-agents are not associated with any specific circuit, but can be interpreted as a barrier between the inputs and the circuits. This helps guarantee that circuits follow the same assignment to the input string.
Alternative agents provide a fallback coalition for each gate agent or assignment agent representing the negation of the assignment of that bit or gate (whereas, as we will see, agents who correspond to their gate's/bit's truth assignment will all be grouped in a single main coalition). 
Notice that the value that a $V$-agent assigns to a corresponding output-gate agent $\ot_j$ dominates the sum of all values she assigns to agents associated with less significant bits of the output. This will be crucial to establish that the strong popularity property of an input string is preserved in the corresponding partition in the constructed FHG, and vice versa.
In the following two subsections, we show that a Condorcet string exists in $\fCC$ if and only if a strongly popular partition exists in the constructed FHG.
Illustrations of the overview of the reduction and its gadgets can be seen in \Cref{fig:Overview,fig:FHG_gadgets} respectively.
\begin{figure}[ht]
\centering
\begin{minipage}{0.45\textwidth}
    \centering
    \begin{tikzpicture}[
    >={Stealth[length=2mm,width=1.3mm]}
]

\node[circle, fill, minimum size=3pt, inner sep=0pt, label=above:$x_i$] (x) {};

\node[circle, fill, minimum size=3pt, inner sep=0pt, label=below:$\assf_i$] (assf)   [below left=1.2cm and 0.6cm of x] {};
\node[circle, fill, minimum size=3pt, inner sep=0pt, label=below:$\asst_i$] (asst)   [below right=1.2cm and 0.6cm of x] {};

\node[color=blue, circle, draw, minimum size=3.5mm, inner sep=0pt, label={[text=gray]below:$m+1$}] (Fa)  [left=1cm of assf] {};
\node[color=blue, circle, draw, minimum size=3.5mm, inner sep=0pt, label={[text=gray]below:$m+1$}] (Ta)  [right=1cm of asst] {};

\node[circle, fill, minimum size=3pt, inner sep=0pt, label={[label distance=-5pt]below left:$\ls_i$}] (la) [below=1cm of $(assf)!0.5!(asst)$] {};

\node[color=blue, circle, draw, minimum size=3.5mm, inner sep=0pt, label={[text=gray]right:$m+1$}] (La)  [below=1cm of la] {};

\draw[->, shorten >=0.8cm, shorten <=0cm] (asst) -- (x);
\draw[dashed, gray, ->, shorten >=0.55cm, shorten <=0cm] (x) -- (asst);
\draw[->, shorten >=0.8cm, shorten <=0cm] (assf) -- (x);
\draw[dashed, gray, ->, shorten >=0.55cm, shorten <=0cm] (x) -- (assf);
\draw (asst) -- (la);
\draw (assf) -- (la);
\draw[blue, thick] (asst) -- (Ta);
\draw[blue, thick] (assf) -- (Fa);
\draw[blue, thick] (la) -- (La);

\end{tikzpicture}
    \caption{Assignment gadget of bit $\XX_i$.}
    \label{FHG_fig:A-gadget}
\end{minipage}
\hfill
\begin{minipage}{0.45\textwidth}
    \centering
    \begin{tikzpicture}[
    every node/.style={inner sep=0pt},
    >={Stealth[length=2mm,width=1.3mm]}
]

\node[circle, fill, minimum size=3pt, inner sep=0pt, label=left:$\assf_i$] (assf) {};
\node[circle, fill, minimum size=3pt, inner sep=0pt, label=right:$\asst_i$] (asst) [right=1.2cm of assf] {};

\node[circle, fill, minimum size=3pt, inner sep=0pt, label=below left:$\gf_i$] (gf)   [below=1.3cm of assf] {};
\node[circle, fill, minimum size=3pt, inner sep=0pt, label=below right:$\gt_i$] (gt)   [below=1.3cm of asst] {};

\node[color=blue, circle, draw, minimum size=3.5mm, inner sep=0pt, label={[text=gray, label distance=3pt]below:$m+1$}] (ggf)  [left=1cm of gf] {};
\node[color=blue, circle, draw, minimum size=3.5mm, inner sep=0pt, label={[text=gray, label distance=3pt]below:$m+1$}] (ggt)  [right=1cm of gt] {};

\node[circle, fill, minimum size=3pt, inner sep=0pt, label=below left:$\lg_i$] (lg)   [below=1cm of $(gf)!0.5!(gt)$] {};

\node[color=blue, circle, draw, minimum size=3.5mm, inner sep=0pt, label={[text=gray, label distance=3pt]right:$m+1$}] (lgg)  [below=1cm of lg] {};

\draw[->, shorten >=0.8cm, shorten <=0cm] (gt) -- (asst);
\draw[dashed, gray,->, shorten >=0.5cm, shorten <=0cm] (asst) -- (gt);
\draw (gt) -- (lg);
\draw[->, shorten >=0.8cm, shorten <=0cm] (gf) -- (assf);
\draw[dashed, gray,->, shorten >=0.5cm, shorten <=0cm] (assf) -- (gf);
\draw (gf) -- (lg);
\draw[blue, thick] (gt) -- (ggt);
\draw[blue, thick] (gf) -- (ggf);
\draw[blue, thick] (lg) -- (lgg);

\end{tikzpicture}
    \caption{Copy-gadget of gate $\GG_i=\XX_i$.}
    \label{FHG_fig:Copy-gadget}
\end{minipage}

\vspace{1em} 

\begin{minipage}{0.45\textwidth}
    \centering
    \begin{tikzpicture}[
    every node/.style={inner sep=0pt},
    >={Stealth[length=2mm,width=1.3mm]}
]

\node[circle, fill, minimum size=3pt, inner sep=0pt, label=left:$\gf_j$] (gfj) {};
\node[circle, fill, minimum size=3pt, inner sep=0pt, label=right:$\gt_j$] (gtj) [right=1.2cm of gfj] {};

\node[circle, fill, minimum size=3pt, inner sep=0pt, label=below left:$\gf_i$] (gf)   [below=1.3cm of gfj] {};
\node[circle, fill, minimum size=3pt, inner sep=0pt, label=below right:$\gt_i$] (gt)   [below=1.3cm of gtj] {};

\node[color=blue, circle, draw, minimum size=3.5mm, inner sep=0pt, label={[text=gray, label distance=3pt]below:$m+1$}] (ggf)  [left=1cm of gf] {};
\node[color=blue, circle, draw, minimum size=3.5mm, inner sep=0pt, label={[text=gray, label distance=3pt]below:$m+1$}] (ggt)  [right=1cm of gt] {};

\node[circle, fill, minimum size=3pt, inner sep=0pt, label=below left:$\lg_i$] (lg)   [below=1cm of $(gf)!0.5!(gt)$] {};

\node[color=blue, circle, draw, minimum size=3.5mm, inner sep=0pt, label={[text=gray, label distance=3pt]right:$m+1$}] (lgg)  [below=1cm of lg] {};

\draw[->, shorten >=1.3cm, shorten <=0cm] (gt) -- (gfj);
\draw[dashed, gray, ->, shorten >=0.55cm, shorten <=0cm] (gfj) -- (gt);
\draw (gt) -- (lg);
\draw[->, shorten >=1.3cm, shorten <=0cm] (gf) -- (gtj);
\draw[dashed, gray, ->, shorten >=0.55cm, shorten <=0cm] (gtj) -- (gf);
\draw (gf) -- (lg);
\draw[blue, thick] (gt) -- (ggt);
\draw[blue, thick] (gf) -- (ggf);
\draw[blue, thick] (lg) -- (lgg);

\end{tikzpicture}
    \caption{Not-gadget of gate $\GG_i=\neg\GG_j$.}
    \label{FHG_fig:Not-gadget}
\end{minipage}
\hfill
\begin{minipage}{0.53\textwidth}
    \centering
    \begin{tikzpicture}[
    every node/.style={inner sep=0pt},
    >={Stealth[length=2mm,width=1.3mm]}
]

\node[circle, fill, minimum size=3pt, inner sep=0pt, label=above:$\gf_j$] (gfj) {};
\node[circle, fill, minimum size=3pt, inner sep=0pt, label=above:$\gf_k$] (gfk) [right=0.8cm of gfj] {};
\node[circle, fill, minimum size=3pt, inner sep=0pt, label=above:$\gt_j$] (gtj) [right=1.2cm of gfk] {};
\node[circle, fill, minimum size=3pt, inner sep=0pt, label=above:$\gt_k$] (gtk) [right=0.8cm of gtj] {};

\node[circle, fill, minimum size=3pt, inner sep=0pt, label={[label distance=3pt]below:$\gf_i$}] (gf)   [below=3.3cm of $(gfj)!0.5!(gfk)$] {};
\node[circle, fill, minimum size=3pt, inner sep=0pt, label={[label distance=3pt]right:$\gt_i$}] (gt)   [below=3.3cm of $(gtj)!0.5!(gtk)$] {};

\node[color=blue, circle, draw, minimum size=3.5mm, inner sep=0pt, label={[text=gray, label distance=3pt]left:$m+1$}] (ggf) [below left=1.1cm and 0.7cm of gf] {};

\node[color=blue, circle, draw, minimum size=3.5mm, inner sep=0pt, label={[text=gray, label distance=3pt]left:$2m+3$}] (ggt) [below right=1.1cm and 0.7cm of gt] {};

\node[circle, fill, minimum size=3pt, inner sep=0pt, label={[label distance=3pt]above:$\lg_i$}] (lg)   [below=0.6cm of $(gf)!0.5!(gt)$] {};

\node[color=blue, circle, draw, minimum size=3.5mm, inner sep=0pt, label={[text=gray, label distance=3pt]left:$m+1$}] (lgg) [below=0.4cm of lg] {};

\node[circle, fill, minimum size=3pt, inner sep=0pt, label=above left:$\wand_i$] (w)   [left=1.8cm of $(gf)!0.27!(gfj)$] {};
\node[circle, fill, minimum size=3pt, inner sep=0pt, label=above left:$\zand_i$] (z)  [below=1.7cm of w] {};
\node[circle, fill, minimum size=3pt, inner sep=0pt, label=above left:$\yand_i$] (y)  [above=1.7cm of w] {};

\node[color=blue, circle, draw, minimum size=2.7mm, inner sep=0pt, label={[text=gray, label distance=3pt]left:$m+1$}] (wrep)  [left=1cm of w] {};
\node[color=blue, circle, draw, minimum size=2.7mm, inner sep=0pt, label={[text=gray, label distance=3pt]left:$m+1$}] (zrep)  [left=1cm of z] {};
\node[color=blue, circle, draw, minimum size=2.7mm, inner sep=0pt, label={[text=gray, label distance=3pt]left:$m+1$}] (yrep)  [left=1cm of y] {};

\node[circle, fill, minimum size=3pt, inner sep=0pt, label=above left:$\lw_i$] (lw)   [above=0.9cm of w] {};
\node[circle, fill, minimum size=3pt, inner sep=0pt, label=above left:$\lz_i$] (lz)  [above=0.9cm of z] {};
\node[circle, fill, minimum size=3pt, inner sep=0pt, label=above left:$\ly_i$] (ly)  [above=0.9cm of y] {};

\node[color=blue, circle, draw, minimum size=2.7mm, inner sep=0pt, label={[text=gray, label distance=3pt]left:$m+1$}] (lwrep)  [left=1cm of lw] {};
\node[color=blue, circle, draw, minimum size=2.7mm, inner sep=0pt, label={[text=gray, label distance=3pt]left:$m+1$}] (lzrep)  [left=1cm of lz] {};
\node[color=blue, circle, draw, minimum size=2.7mm, inner sep=0pt, label={[text=gray, label distance=3pt]left:$m+1$}] (lyrep)  [left=1cm of ly] {};

\draw[blue, thick] (gt) -- (ggt);
\draw[blue, thick] (gf) -- (ggf);
\draw[blue, thick] (w) -- (wrep);
\draw[blue, thick] (z) -- (zrep);
\draw[blue, thick] (y) -- (yrep);
\draw[blue, thick] (lw) -- (lwrep);
\draw[blue, thick] (lz) -- (lzrep);
\draw[blue, thick] (ly) -- (lyrep);
\draw[blue, thick] (lg) -- (lgg);
\draw (gt) -- (lg);
\draw (gf) -- (lg);
\draw[dashed, gray] (gf) -- (gtj);
\draw[dashed, gray] (gf) -- (gtk);
\draw[dashed, gray] (w) -- (gfk);
\draw[dashed, gray] (w) -- (gtj);
\draw[dashed, gray] (z) -- (gfj);
\draw[dashed, gray] (z) -- (gtk);
\draw[dashed, gray] (y) -- (gfj);
\draw[dashed, gray] (y) -- (gfk);

\draw[dashed, gray, ->, shorten >=1.38cm, shorten <=0cm] (gtj) -- (gt);
\draw[cyan, ->, shorten >=1.9cm, shorten <=0cm] (gt) -- (gtj);

\draw[dashed, gray, ->, shorten >=1.38cm, shorten <=0cm] (gtk) -- (gt);
\draw[cyan, ->, shorten >=1.9cm, shorten <=0cm] (gt) -- (gtk);

\draw[dashed, gray, ->, shorten >=1.38cm, shorten <=0cm] (gfj) -- (gf);
\draw[brown, ->, shorten >=1.9cm, shorten <=0cm] (gf) -- (gfj);

\draw[dashed, gray, ->, shorten >=1.38cm, shorten <=0cm] (gfk) -- (gf);
\draw[brown, ->, shorten >=1.9cm, shorten <=0cm] (gf) -- (gfk);

\draw[magenta, ->, shorten >=1.04cm, shorten <=0cm] (gf) -- (w);
\draw[black, ->, shorten >=1.04cm, shorten <=0cm] (w) -- (gf);

\draw[magenta, ->, shorten >=1.04cm, shorten <=0cm] (gf) -- (z);
\draw[black, ->, shorten >=1.04cm, shorten <=0cm] (z) -- (gf);

\draw[cyan, ->, shorten >=1.9cm, shorten <=0cm] (gf) -- (y);
\draw[black, ->, shorten >=1.38cm, shorten <=0cm] (y) -- (gf);

\draw[->, shorten >=0.54cm, shorten <=0cm] (w) -- (lw);
\draw[dashed, gray, ->, shorten >=0.36cm, shorten <=0cm] (lw) -- (w);

\draw[->, shorten >=0.54cm, shorten <=0cm] (z) -- (lz);
\draw[dashed, gray, ->, shorten >=0.36cm, shorten <=0cm] (lz) -- (z);

\draw[->, shorten >=0.54cm, shorten <=0cm] (y) -- (ly);
\draw[dashed, gray, ->, shorten >=0.36cm, shorten <=0cm] (ly) -- (y);

\end{tikzpicture}
    \caption{And-gadget of gate $\GG_i=\GG_j\land \GG_k$. $\lw$, $\lz$, and $\ly$ additionally assign $2$ to one of their replicas.}
    \label{FHG_fig:And-gadget}
\end{minipage}

\caption{The main gadgets of the proof of \Cref{thm:FHG}, depicted along with their input agents. 
Black dots represent single agents. 
Blue circles represent sets of one-way replicas, with a blue edge to their origin;
valuations concerning them are not depicted, and can be derived from their origin.
Non-blue edge colors indicate valuation: black = $2$, magenta=$1.75$, cyan = $1.5$, brown=$1.25$, and dashed gray = $1$.
Undirected edges imply mutual valuation of the respective color.
Omitted edges imply mutual valuation of $0$.}
\label{fig:FHG_gadgets}
\end{figure}

We conclude this section with two definitions which formalize the connection between the coalitions of the agents in the FHG and the values of the input bits and gates of the Boolean circuits.

\begin{definition}
\label{def_FHG:valid_partition}
{\bf (FHG $\pi$-correspondence, $\pi$-validity).}
Let $\XX_i$ be an input bit with corresponding $A$-gadget $\ADD_i$, let $\pi$ be a partition of $N$, and let $\MC:=\pi(\xx_1)$. 
If $\asst_i\in \MC$ and $\assf_i\notin \MC$ then we say that $\ADD_i$ \emph{$\pi$-corresponds to assignment $\XX_i=1$}, 
and if $\assf_i\in \MC$ and $\asst_i\notin \MC$ then we say that $\ADD_i$ \emph{$\pi$-corresponds to assignment $\XX_i=0$}.
If $\ADD_i$ $\pi$-corresponds to either $0$ or $1$, we say it is a \emph{$\pi$-valid gadget}.

Let $\GG_i$ be a gate. If $\gt_i\in \MC$ and $\gf_i\notin \MC$, we say that $\GDD_i$ \emph{$\pi$-corresponds to value $\GG_i=1$}, and if $\gf_i\in \MC$ and $\gt_i\notin \MC$, we say that $\GDD_i$ \emph{$\pi$-corresponds to value $\GG_i=0$}.
If $\GDD_i$ $\pi$-corresponds to either $0$ or $1$, we say it is a \emph{$\pi$-valid gadget}.
\end{definition}

\begin{definition}
\label{def_FHG:compliance}
{\bf (FHG $\pi$-compliance).}
Let $\GG_i$ be a gate with corresponding $G$-gadget $\GDD_i$, let $\pi$ be a partition of $N$, and let $\MC:=\pi(\xx_1)$. 
Suppose the input gadget(s) of $\GG_i$ are $\pi$-valid.
We say that $\GDD_i$ \emph{$\pi$-complies with its inputs} if $\GDD_i$ $\pi$-corresponds to the value that $\GG_i$ obtains when evaluated on the values its input gadgets $\pi$-correspond to (e.g., if $\GG_i=\neg\GG_j$ and $\GDD_j$ $\pi$-corresponds to $\GG_j=1$, then $\GDD_i$ $\pi$-complies with its input if it $\pi$-corresponds to $\GG_i=0$).
\end{definition}

\subsection{Correctness of the Reduction}
\label{sec_FHG:correctness}
We proceed to show that the $\CKTCondorcet{}$ instance $\fCC$ admits a Condorcet string if and only if the reduced FHG instance $\langle N,\mathbf{v}\rangle$ described above admits a strongly popular partition.
We begin by proving a simple lemma.

\begin{lemma}
\label{lem_FHG:X_R_maximal_util_is_1}
The maximal utility any $\xx\in X$ or any replica agent $r$ may obtain in any partition is $1$.
\end{lemma}
\begin{proof}
Agent $\xx$ assigns value $v^x$ to any other agent in $X\cup V\cup Q$, $0$ to any $L$-agent, and $1$ to all other agents.
In particular, $\mv{\xx}=v^x$.
Thus, by \Cref{prop_FHG:maximal_util_welcome} any partition that yields maximum utility for $\xx$ must include all of $X\cup V$.
The claim then follows: 
for the coalition $X\cup V\cup Q$, we have $u_{\xx}(X\cup V\cup Q)=\frac{(|X\cup V\cup Q|-1)v^x}{|X\cup V\cup Q|}=1$, and then, since $\xx$ assigns either $1$ or $0$ to any agent outside $X\cup V\cup Q$, for any coalition $S\subseteq N\bs (X\cup V\cup Q)$ we have:
\[u_{\xx}(X\cup V\cup Q\cup S)=\frac{(|X\cup V\cup Q|-1)v^x+\countv{x}{1}{\pi}}{|X\cup V\cup Q|+\countv{x}{1}{\pi}+\countv{x}{0}{\pi}}=
\]
\[
\frac{|X\cup V\cup Q|+\countv{x}{1}{\pi}}{|X\cup V\cup Q|+\countv{x}{1}{\pi}+\countv{x}{0}{\pi}}\le 1.
\]

The analysis for agent $r$ is analogous.
\end{proof}

We proceed with the introduction of a useful definition, followed by a lemma concerning it.
 
\begin{definition}
\label{def_FHG:natural_partition}
\textbf{(FHG Natural Partition).}
For any given bit string $\XX\in \bits{n}$, we let $\pi_{\XX}$ denote the following partition, also referred to as the \emph{natural partition} corresponding to $\XX$.
\begin{itemize}
    \item All agents in $X\cup V\cup Q$ are in a coalition denoted $\MC_{\XX}$ (for denoting a \emph{main coalition} to which we will add further agents).
    \item For every input bit $\XX_i$, if $\XX_i=1$ then $\Ta_i\subseteq \MC_{\XX}$ and $\Fa_i\cup\La_i\in \pi_{\XX}$, and if $\XX_i=0$ then $\Fa_i\subseteq \MC_{\XX}$ and $\Ta_i\cup\La_i\in \pi_{\XX}$.
    \item For every gate $\GG_i$, if $\XX(\GG_i)=1$ then $\Tg_i\subseteq \MC_{\XX}$ and $\Fg_i\cup\Lg_i\in\pi_{\XX}$, and if $\XX(\GG_i)=0$ then $\Fg_i\subseteq\MC_{\XX}$ and $\Tg_i\cup\Lg_i\in\pi_{\XX}$.
    \item For every And-gate $\GG_i=\GG_j\land \GG_k$:
    \begin{itemize}
        \item if $\XX(\GG_j)=\XX(\GG_k)=1$ then $\{\Wand_i\cup\Lw_i,\Zand_i\cup\Lz_i,\Yand_i\cup\Ly_i\}\subseteq\pi_{\XX}$.
        \item if $\XX(\GG_j)=\XX(\GG_k)=0$ then $\{\Wand_i\cup\Lw_i, \Zand_i\cup\Lz_i, \Ly_i\}\subseteq\pi_{\XX}$, and $\Yand_i\subseteq \MC_{\XX}$.
        \item if $\XX(\GG_j)=0$ and $\XX(\GG_k)=1$ then $\{\Wand_i\cup\Lw_i, \Lz_i, \Yand_i\cup\Ly_i\}\subseteq\pi_{\XX}$, and $\Zand_i\subseteq\MC_{\XX}$.
        \item if $\XX(\GG_j)=1$ and $\XX(\GG_k)=0$ then $\{\Lw_i, \Zand_i\cup\Lz_i,\Yand_i\cup\Ly_i\}\subseteq\pi_{\XX}$, and $\Wand_i\subseteq\MC_{\XX}$.
    \end{itemize}
\end{itemize}
\end{definition}

\begin{lemma}
\label{lem_FHG:natural_utility_for_non_voters_1}
For any $\XX\in \bits{n}$, all agents in $N\bs V$ obtain utility $1$ in $\pi_{\XX}$.
\end{lemma}
\begin{proof}
Core agents from non-And-gadgets have in their coalition one agent to whom they assign value $2$, and some agents to whom they assign $1$. 
Core agents $\gt_i$ and $\gf_i$ within an And-gadget $\GDD_i$ may have a slightly different composition of agents in their coalitions (though always achieving utility $1$). 
The utility of $X$-agents is discussed in \Cref{lem_FHG:X_R_maximal_util_is_1,lem_FHG:X_R_margin}.
For $Q$-agents, their coalition $\MC_{\XX}$ contains exactly half of $(G\cup A)\bs R$, and some agents to whom they assign $1$; therefore $\forall q\in Q\;u_q(\pi_{\XX})=\frac{\frac{1}{2}|(G\cup A)\bs R|v^q+\countv{q}{1}{\pi_{\XX}}}{1+\frac{1}{2}|(G\cup A)\bs R|+\countv{q}{1}{\pi_{\XX}}}=1$.
Lastly, any replica agent $r\in r(\ga)$ of some origin $\ga$ is in a coalition with all other agents of $r(\ga)\cup\{\ga\}$, and has $\countv{r}{0}{\pi_{\XX}}=0$.
Thus, by definition of replicas we have $u_r(\pi_{\XX})=\frac{|r(\ga)|\frac{|r(\ga)|+1}{|r(\ga)|}+\countv{r}{1}{\pi_{\XX}}}{1+r(\ga)+\countv{r}{1}{\pi_{\XX}}}=1$.    
\end{proof}

Our next goal is to show the following claim.
\begin{lemma}
\label{lem_FHG:natural_partitions_beat_others}
Let $\XX\in \bits{n}$ be an arbitrary string, and let $\pi$ be a Pareto-optimal partition such that $\phi(\pi_{\XX},\pi)\le 0$. Then $\pi=\pi_{\XX'}$ for some string $\XX'$.
\end{lemma}

Proving \Cref{lem_FHG:natural_partitions_beat_others} requires several auxiliary lemmas. 
Once it is proved, the proof of \Cref{thm:FHG} will follow without much effort.
Let $\XX\in\bits{n}$ and let $\pi$ be a Pareto-optimal partition with $\phi(\pi_{\XX},\pi)\le 0$. 
Denote $\MC=\pi(\xx_1)$ (we will see that $\MC$ will be the ``main'' coalition in $\pi$, similarly to $\MC_{\XX}$ in $\pi_{\XX}$).
We establish several separate claims concerning the structure of $\pi$, and the popularity margins of different agent sets in the vote of $\pi$ against $\pi_{\XX}$.

\begin{lemma}
\label{lem_FHG:X_R_margin}
Let $\ga\in X\cup R$.
Then $\ga$ obtains maximal utility in $\pi_{\XX}$.
In particular, we have $\phi_{\ga}(\pi_{\XX},\pi)\geq 0$.
\end{lemma}
\begin{proof}
First, assume $\ga=\xx\in X$.
Since $X\cup V\cup Q\subseteq\pi_{\XX}(\xx)$, and $L\cap \pi_{\XX}(\xx)=\emptyset$, we have that $u_{\xx}(\pi_{\XX})=1$ (by a similar argument as in \Cref{lem_FHG:X_R_maximal_util_is_1}).
By \Cref{lem_FHG:X_R_maximal_util_is_1}, this is the maximal utility obtainable by $\xx$, and in particular $\xx$ cannot obtain higher utility in $\pi$.

Second, assume $\ga=r\in r(\ga')$ for some origin agent $\ga'$.
Since $r(\ga')\subseteq\pi_{\XX}(\ga')$, and $r$ assigns $1$ to any other agent in $\pi_{\XX}(r)\bs (\{\ga'\}\cup r(\ga'))$, we have that $u_{r}(\pi_{\XX})=1$ (with an analysis similar to the proof of \Cref{lem_FHG:X_R_maximal_util_is_1}).
By \Cref{lem_FHG:X_R_maximal_util_is_1}, this is the maximal utility obtainable by $r$, and in particular $r$ cannot obtain higher utility in $\pi$.
\end{proof}

\begin{lemma}
\label{lem_FHG:MC_X_V_Q}
It holds that $X\cup V\cup Q\subseteq \MC$.
\end{lemma}
\begin{proof}
Assume for contradiction that there exists some $\ga\in (X\bs\{\xx_1\})\cup V\cup Q$ such that $\ga\notin\pi(\xx_1)=MC$. 
Then, by \Cref{prop_FHG:maximal_util_welcome}, no agent $\xx'\in X$ obtains maximal utility in $\pi$ (since there is an agent outside $\pi(\xx')$ to whom $\xx'$ assigns value $\mv{\xx'}$).
Thus, by \Cref{lem_FHG:X_R_margin} all $X$-agents prefer $\pi_{\XX}$ over $\pi$, and so $\phi_{X}(\pi_{\XX},\pi) = |X|$. 
By design, $|X|>\frac{|N|}{2}$, and therefore $\pi_{\XX}$ is more popular than $\pi$, a contradiction. 
\end{proof}

\begin{lemma}
\label{lem_FHG:L_notin_MC}
It holds that $L\cap \MC=\emptyset$.
\end{lemma}
\begin{proof}
Assume for contradiction that $L\cap \MC\ne\emptyset$.
Then $\countv{\xx}{0}{\pi}>0$ for all $\xx\in X$.
By \Cref{lem_FHG:MC_X_V_Q}, we have that $\pi(\xx)$ consists of all $X$-, $V$-, and $Q$-agents, some $L$-agents, and possibly some agents to which $\xx$ assigns value $1$.
Then we have
\[u_{\xx}(\pi)=\frac{(|X\cup V\cup Q|-1)v^x+\countv{\xx}{1}{\pi}}{|X\cup V\cup Q|+\countv{\xx}{1}{\pi}+\countv{\xx}{0}{\pi}}=\frac{|X\cup V\cup Q|+\countv{\xx}{1}{\pi}}{|X\cup V\cup Q|+\countv{\xx}{1}{\pi}+\countv{\xx}{0}{\pi}}<1.\]
Thus, by \Cref{lem_FHG:X_R_maximal_util_is_1,lem_FHG:X_R_margin} all $X$-agents prefer $\pi_{\XX}$ over $\pi$, and so $\phi_{X}(\pi_{\XX},\pi) = |X|$. 
By design, $|X|>\frac{|N|}{2}$, and therefore $\pi_{\XX}$ is more popular than $\pi$, a contradiction.
\end{proof}

\begin{lemma}
\label{lem_FHG:A_most_prefer_natural}
Let $\ADD_i$ be an assignment gadget.
Let $\ga\in\ADD_i\bs L$.
Then $\phi_{\ga}(\pi_{\XX},\pi)\ge 0$.
\end{lemma}
\begin{proof}
By \Cref{lem_FHG:natural_utility_for_non_voters_1}, agent $\ga$ obtains utility $1$ in $\pi_{\XX}$.
If $\ga$ is a replica agent, the result follows from \Cref{lem_FHG:X_R_maximal_util_is_1,lem_FHG:natural_utility_for_non_voters_1}.
If $\ga$ is a core agent then by \Cref{lem_FHG:MC_X_V_Q,lem_FHG:L_notin_MC} we have that $\pi(\ls_i)\ne\pi(\xx_i)$; thus, $\ga$ has at most one agent in her coalition in $\pi$ to whom she assigns value $2$. Furthermore, $\ga$ assigns at most $1$ to any other agent, and thus she cannot obtain a utility higher than $1$ in $\pi$.
\end{proof}

\begin{lemma}
\label{lem_FHG:R_bad_coalition}
Let $\ga\in N$ with $r(\ga)\ne \emptyset$. If $r(\ga)\nsubseteq\pi(\ga)$ or $\countv{\ga}{0}{\pi}\ge 1$ then $\phi_{r(\ga)}(\pi_{\XX},\pi)=|r(\ga)|$.
\end{lemma}
\begin{proof}
Let $r\in r(\ga)$. Since $r$ assigns $\mv{r}$ to any other agent in $r(\ga)$, if $r(\ga)\nsubseteq\pi(\ga)$ then $r$ cannot obtain maximal utility in $\pi$ by \Cref{prop_FHG:maximal_util_welcome}.
The claim then follows from \Cref{lem_FHG:X_R_margin}.

If $r(\ga)\subseteq\pi(\ga)$ but $\countv{\ga}{0}{\pi}\ge 1$, then for all $r\in r(\ga)$ we have $\countv{r}{0}{\pi}\ge 1$, yielding
\[u_r(\pi)=\frac{|r(\ga)|\frac{|r(\ga)|+1}{|r(\ga)|}+\countv{r}{1}{\pi}}{1+r(\ga)+\countv{r}{1}{\pi}+\countv{r}{0}{\pi}}=
\frac{|r(\ga)|+1+\countv{r}{1}{\pi}}{1+r(\ga)+\countv{r}{1}{\pi}+\countv{r}{0}{\pi}}<1.\]
The claim then follows from \Cref{lem_FHG:natural_utility_for_non_voters_1}.
\end{proof}

\begin{lemma}
\label{lem_FHG:Lw_Lz_Ly}
Let $\ga\in \Lw\cup\Lz\cup\Ly$ be a non-replica agent.
Then $\phi_{\ga}(\pi_{\XX},\pi)\ge 0$.
Furthermore, if $\countv{\ga}{0}{\pi}\ge 1$, then $\phi_{\ga}(\pi_{\XX},\pi)>0$.
\end{lemma}
\begin{proof}
By \Cref{lem_FHG:natural_utility_for_non_voters_1}, $\ga$ obtains utility $1$ in $\pi_{\XX}$.
This utility is maximal because $\ga$ assigns $2$ to only one of her replicas, and at most $1$ to any other agent.
This proves that $\phi_{\ga}(\pi_{\XX},\pi)\ge 0$.
For the second claim, by the valuations of $\ga$ discussed above and by the assumption that $\countv{\ga}{0}{\pi}\ge 1$, we obtain $u_{\ga}(\pi)=\frac{2\cdot\countv{\ga}{2}{\pi}+\countv{\ga}{1}{\pi}}{1+\countv{\ga}{2}{\pi}+\countv{\ga}{1}{\pi}+\countv{\ga}{0}{\pi}}\le \frac{2\cdot\countv{\ga}{2}{\pi}+\countv{\ga}{1}{\pi}}{2+\countv{\ga}{2}{\pi}+\countv{\ga}{1}{\pi}}$.
It can be seen that for $\countv{\ga}{2}{\pi}\in\{0,1\}$, this quantity is less than $1$.
\end{proof}

\begin{lemma}
\label{lem_FHG:wzy_max_util_with_bad_agent}
Let $\GDD_i$ be an And-gadget, and let $\ga\in\{\wand_i,\zand_i,\yand_i\}$.
Assume $\countv{\ga}{0}{\pi}\ge 1$.
Then $u_{\ga}(\pi)\le 1$.
\end{lemma}
\begin{proof}
Agent $\ga$ assigns value $2$ only to two other agents (her alternative agent, and $\gf_i$), and $1$ or $0$ to any other agent.
Hence, we have that $\countv{\ga}{2}{\pi}\le 2$ while $\countv{\ga}{0}{\pi}\ge 1$.
Thus, we have
\[u_{\ga}(\pi)=
\frac{2\cdot \countv{\ga}{2}{\pi}+\countv{\ga}{1}{\pi}}{1+\countv{\ga}{2}{\pi}+\countv{\ga}{1}{\pi}+\countv{\ga}{0}{\pi}}\le 
\frac{2\cdot \countv{\ga}{2}{\pi}+\countv{\ga}{1}{\pi}}{2+\countv{\ga}{2}{\pi}+\countv{\ga}{1}{\pi}}.\]
It can be seen that for $\countv{\ga}{2}{\pi}\in\{0,1,2\}$, this quantity is at most $1$.
\end{proof}

\begin{lemma}
\label{lem_FHG:agent_improvement_ruins_gadget}
Let $\DD_i$ be an assignment or gate gadget, and let $\ga\in \DD_i$.
Assume $\phi_{\ga}(\pi_{\XX},\pi)<0$. 
Then $\phi_{\DD_i}(\pi_{\XX},\pi)\ge m+1$.
\end{lemma}
\begin{proof}
By \Cref{lem_FHG:X_R_margin,lem_FHG:Lw_Lz_Ly,lem_FHG:A_most_prefer_natural}, the only agents in $\DD_i$ who may prefer $\pi$ over $\pi_{\XX}$ are:
\begin{enumerate}
    \item the $\Lg$- or $\La$-agent of the gadget (excluding replicas);
    \item the core agents, in the case that $\DD_i$ is a gate gadget; and
    \item agents $\wand_i$, $\zand_i$, and $\yand_i$, in the case that $\DD_i$ is an And-gadget.
\end{enumerate}
In particular, $\ga$ must be one of those agents, and must obtain utility more than $1$ in $\pi$ by \Cref{lem_FHG:natural_utility_for_non_voters_1}.

First, consider the case that the two core agents of $\DD_i$ are in the same coalition.
Then by \Cref{lem_FHG:R_bad_coalition} any replica $r$ of the core agents of $\DD_i$ must strictly prefer $\pi_{\XX}$:
in $\pi$, either $r$ is not in the coalition of their origin or they are in a coalition with their origin but also with the complementary core agent.
Furthermore, agents $\wand_i$, $\zand_i$, and $\yand_i$ (if $\DD_i$ is an And-gadget) cannot prefer $\pi$ over $\pi_{\XX}$: either they are in the joint coalition of $\gt_i$ and $\gf_i$, in which case we can apply \Cref{lem_FHG:wzy_max_util_with_bad_agent}, or not, in which case they assign at most $1$ to anyone in their coalition apart from, possibly, their respective alternative agent, yielding utility at most $1$ for them.
It can also be seen that agents $\gt_i$ and $\asst_i$ of any type of gadget cannot obtain utility higher than $1$ if they assign $0$ to anyone in their coalition (by considering all combinations of agents to whom they assign more than $1$).
Thus, only two agents in $\DD_i$ may prefer $\pi$ over $\pi_{\XX}$, while at least the $2m+2$ replicas of the core agents prefer $\pi_{\XX}$, yielding $\phi_{\DD_i}(\pi_{\XX},\pi)\ge 2m+2-2\ge m+1$, and we are done.

Thus, the core agents of $\DD_i$ are in separate coalitions of $\pi$.
It then follows that the alternative agent belonging to $\DD_i\cap(\Lg\cup\La)$ cannot prefer $\pi$ over $\pi_{\XX}$:
she assigns value $2$ to those core agents, and $1$ or $0$ to any other agent, and since she cannot have both core agents in her coalition, she obtains utility at most $1$.

Hence, $\ga$ must be a core agent of a gate gadget, or an agent from $\{\wand_i,\zand_i,\yand_i\}$ of an And-gadget. 
We distinguish cases according to the type of $\DD_i$.
First, assume that $\DD_i$ is a Not- or Copy-gadget.
We assume that $\ga=\gt_i$, as the proof is symmetric in the case that $\ga=\gf_i$. 
Since $u_{\gt_i}(\pi)>1$ we have that $\pi(\gt_i)$ contains $\lg_i$, and also one of the agents corresponding to the inputs of $\DD_i$ ($\asst_i$ if $\DD_i$ corresponds to a Copy-gate $\GG_i=\XX_i$, and $\gf_j$ if $\DD_i$ corresponds to a Not-gate $\GG_i=\neg\GG_j$).
By \Cref{lem_FHG:R_bad_coalition} we have that $\phi_{\Lg_i}(\pi_{\XX},\pi)=|\Lg_i|=m+2$.
Furthermore, since we established that the core agents of $\DD_i$ must be in separate coalitions in $\pi$, we have that $\lg_i\notin\pi(\gf_i)$, implying that $\gf_i$ cannot prefer $\pi$ over $\pi_{\XX}$. Thus, within $\DD_i$ only $\gt_i$ may prefer $\pi$, yielding $\phi_{\DD_i}\ge (m+2)-1=m+1$.

Thus, $\DD_i$ corresponds to an And-gate $\GG_i=\GG_j\land\GG_k$, and $\ga\in\{\gt_i,\gf_i,\wand_i,\zand_i,\yand_i\}$.

If $\ga\in\{\wand_i,\zand_i,\yand_i\}$, let us assume that $\ga=\wand_i$ (the cases $\ga=\zand_i$ and $\ga=\yand_i$ are analogous). Then, $\{\lw_i,\gf_i\}\subseteq\pi(\wand_i)$.
It follows that all agents in $r(\gf_i)\cup\Lw_i$ are either in $\pi(\wand_i)$, and thus \Cref{lem_FHG:R_bad_coalition} or \ref{lem_FHG:Lw_Lz_Ly} can be applied to them, or they are not, in which case they must be a replica agent separated from its origin; either way, these agents must prefer $\pi_{\XX}$ over $\pi$.
Furthermore, since $\gf_i\in\pi(\wand_i)$ we have that $\zand_i$ and $\yand_i$ cannot prefer $\pi$ over $\pi_{\XX}$, since either they are in $\pi(\wand_i)$ and we can apply \Cref{lem_FHG:wzy_max_util_with_bad_agent}, or they are not and thus they assign value $2$ to at most one agent in their coalition and value at most $1$ to any other agent.
Thus, we have at least $2m+3$ agents preferring $\pi_{\XX}$ and at most three preferring $\pi$, and we are done.

Thus, we must have that $\ga\in\{\gt_i,\gf_i\}$ while all other agents cannot prefer $\pi$ over $\pi_{\XX}$, by the previous cases. 
We observe the following.
\begin{enumerate}
    \item In order to have $u_{\gt_i}(\pi)>1$, $\pi(\gt_i)$ must contain $\lg_i$ and at least one of $\{\gt_j,\gt_k\}$.\label{lem_FHG:agent_improvement_ruins_gadget:item_gt}
    \item In order to have $u_{\gf_i}(\pi)>1$, at least one of the following must hold:
    \begin{enumerate}
        \item $|\{\wand_i,\zand_i,\yand_i\}\cap\pi(\gf_i)|\ge 2$, or\label{lem_FHG:agent_improvement_ruins_gadget:item_gf_wzy2}
        \item $|\{\wand_i,\zand_i\}\cap\pi(\gf_i)|=1$ and $\{\gf_j,\gf_k\}\subseteq \pi(\gf_i)$, or\label{lem_FHG:agent_improvement_ruins_gadget:item_gf_wzjk}
        \item $|\{\wand_i,\zand_i,\yand_i,\gf_j,\gf_k\}\cap\pi(\gf_i)|\ge 1$ and $\lg_i\in \pi(\gf_i)$. \label{lem_FHG:agent_improvement_ruins_gadget:item_gf_lg}
    \end{enumerate}
\end{enumerate}

If \Cref{lem_FHG:agent_improvement_ruins_gadget:item_gf_wzy2} holds then by \Cref{lem_FHG:R_bad_coalition}, all replicas of at least two of $\wand_i$, $\zand_i$, and $\yand_i$ must prefer $\pi_{\XX}$ over $\pi$.
Thus, we have at least $2m+2$ agents preferring $\pi_{\XX}$ against at most two agents preferring $\pi$, and so $\phi_{\DD_i}(\pi_{\XX},\pi)\ge m+1$ as required.
So assume \Cref{lem_FHG:agent_improvement_ruins_gadget:item_gf_wzy2} does not hold.

If \Cref{lem_FHG:agent_improvement_ruins_gadget:item_gt} or \ref{lem_FHG:agent_improvement_ruins_gadget:item_gf_lg} holds, then by \Cref{lem_FHG:R_bad_coalition}, all agents in $r(\lg_i)$ prefer $\pi_{\XX}$.
Furthermore, we have $\countv{\lg_i}{0}{\pi}\ge 1$ and, since the core agents of $\DD_i$ are separated, we have $\countv{\lg_i}{2}{\pi}\le 1$, while $\ell$ assigns at most $1$ to any other agent. 
Hence, $u_{\lg_i}(\pi)<1$.
Therefore, we have $m+2$ agents preferring $\pi_{\XX}$ over $\pi$.
If only one of $\gt_i$ and $\gf_i$ prefers $\pi$, this suffices to reach the required margin. 
If both of them do, then we must have that \Cref{lem_FHG:agent_improvement_ruins_gadget:item_gt,lem_FHG:agent_improvement_ruins_gadget:item_gf_wzjk} both hold (\Cref{lem_FHG:agent_improvement_ruins_gadget:item_gt,lem_FHG:agent_improvement_ruins_gadget:item_gf_lg} are mutually exclusive since we have shown that the core agents are separated).
Thus, by \Cref{lem_FHG:R_bad_coalition} all replicas of at least one of $\wand_i$ or $\zand_i$ must prefer $\pi_{\XX}$, and we are done.

The preceding arguments rule out \Cref{lem_FHG:agent_improvement_ruins_gadget:item_gt,lem_FHG:agent_improvement_ruins_gadget:item_gf_wzy2,lem_FHG:agent_improvement_ruins_gadget:item_gf_lg}, which implies that only $\gf_i$ prefers $\pi$ (and so $\ga=\gf_i$), and only \Cref{lem_FHG:agent_improvement_ruins_gadget:item_gf_wzjk} holds.
Thus, $\pi(\gf_i)$ contains exactly one agent $\ga'\in\{\wand_i,\zand_i\}\cap \pi(\gf_i)$, both $\gf_j$ and $\gf_k$, and apart from that only agents to whom $\gf_i$ assigns value $1$ or $0$.
Therefore, we must have $\countv{\gf_i}{0}{\pi}=0$ as otherwise $\gf_i$ does not prefer $\pi$ over $\pi_{\XX}$, a contradiction.
In particular, this implies that $\{\lw_i,\lz_i\}\cap\pi(\gf_i)=\emptyset$.
Thus, $\ga'$ must prefer $\pi_{\XX}$, and by \Cref{lem_FHG:R_bad_coalition} so do all her replicas $r(\ga')$, which sums to $m+2$ agents.
Conversely, only one agent in $\DD_i$ prefers $\pi$, concluding the proof.
\end{proof}

\begin{lemma}
\label{lem_FHG:A_G_margin}
Let $\DD_i$ be an assignment or gate gadget.
Then $\phi_{\DD_i}(\pi_{\XX},\pi)\ge 0$.
\end{lemma}
\begin{proof}
Either no agent in $\DD_i$ prefers $\pi$ over $\pi_{\XX}$, or else the result follows from \Cref{lem_FHG:agent_improvement_ruins_gadget}.
\end{proof}

We now introduce some terminology to reason about several groups of agents which must be separated from one another in $\pi$.
We call two sets $S_1,S_2\subseteq N$ \emph{clashing} if $\forall (s_1,s_2)\in S_1\times S_2\; \mathbf{v}_{s_1}(s_2)=\mathbf{v}_{s_2}(s_1)=0$.
Furthermore, for an assignment or gate gadget $\DD_i$, let $\mathcal{S}_i=\{T_i,F_i,\Wand_i,\Zand_i,\Yand_i\}$ if $\DD_i$ is an And-gadget, and let $\mathcal{S}_i=\{T_i,F_i\}$ otherwise.
Thus, all pairs of sets from $\mathcal{S}_i$ are clashing apart from pairs $(F_i,S)$ where $S\in\{\Wand_i,\Zand_i,\Yand_i\}$.

\begin{lemma}
\label{lem_FHG:clashing_sets}
Let $\DD_i$ and $\DD_j$ be assignment or gate gadgets (possibly $i=j$), and let $S_i\in \mathcal{S}_i$ and $S_j\in \mathcal{S}_j$ be two clashing sets. 
Then for every $s_i\in S_i$ and $s_j\in S_j$ we have $\pi(s_i)\neq\pi(s_j)$.
\end{lemma}
\begin{proof}
Assume the contrary.
Then by \Cref{lem_FHG:R_bad_coalition} all replica agents in $S_i\cup S_j$ prefer $\pi_{\XX}$ over $\pi$, and so by \Cref{lem_FHG:X_R_margin,lem_FHG:A_G_margin} we have 
\[\phi(\pi_{\XX},\pi)=\phi_X(\pi_{\XX},\pi)+\phi_V(\pi_{\XX},\pi)+\phi_Q(\pi_{\XX},\pi)+\sum_{\DD_j}\phi_{\DD_j}(\pi_{\XX},\pi)\ge -m-(m+1)+(2m+2)\ge 1,\]
a contradiction.
\end{proof}

\begin{lemma}
\label{lem_FHG:non_voters_prefer_pi*}
Let $\ga\in N\bs V$. Then $\phi_{\ga}(\pi_{\XX},\pi)\ge 0$.
\end{lemma}
\begin{proof}
If $\ga\in X$ the result holds by \Cref{lem_FHG:X_R_margin}.

Assume for contradiction that $\phi_{\ga}(\pi_{\XX},\pi)<0$. 
If $\ga\in Q$, then by \Cref{lem_FHG:natural_utility_for_non_voters_1} we must have $u_{\ga}(\pi)>1$.
Recall that $\ga$ assigns $v^q>1$ to any agent in $(G\cup A)\bs R$, and at most $1$ to any other agent.
Thus, we must have $|\pi(\ga)\cap ((G\cup A)\bs R)|> \frac{1}{2}|(G\cup A)\bs R|$.
Therefore, there must exist a gadget with core-agent sets $T_i$ and $F_i$ such that $\MC\cap T_i\neq\emptyset$ and $\MC\cap F_i\neq\emptyset$, a contradiction to \Cref{lem_FHG:clashing_sets}.

If $\ga\in A\cup G\cup L\cup W\cup Z\cup Y$, let $\DD_i$ be the gadget to which $\ga$ belongs.
Then by \Cref{lem_FHG:agent_improvement_ruins_gadget} we have $\phi_{\DD_i}(\pi_{\XX},\pi)\ge m+1$.
Furthermore, by \Cref{lem_FHG:A_G_margin} we have $\phi_{\DD_j}(\pi_{\XX},\pi)\ge 0$ for any other gadget $\DD_j$, and by \Cref{lem_FHG:X_R_margin} we have $\phi_X(\pi_{\XX},\pi)\ge 0$.
We also have that $\phi_Q(\pi_{\XX},\pi)\ge 0$, by the previous case. 
Thus, even if all $V$-agents prefer $\pi$ over $\pi_{\XX}$ we have that $\pi_{\XX}$ is more popular than $\pi$, a contradiction.
Formally:
\[\phi(\pi_{\XX},\pi)=\phi_X(\pi_{\XX},\pi)+\phi_V(\pi_{\XX},\pi)+\phi_Q(\pi_{\XX},\pi)+\sum_{\DD_j}\phi_{\DD_j}(\pi_{\XX},\pi)\ge -m+(m+1)\ge 1.\qedhere\]
\end{proof}

\begin{lemma}
\label{lem_FHG:replicas_with_origin}
Let $r$ be a replica of some agent $\ga$.
Then $r\in \pi(\ga)$.
\end{lemma}
\begin{proof}
If not, then by \Cref{lem_FHG:R_bad_coalition} all agents in $r(\ga)$ prefer $\pi_{\XX}$. Thus, by \Cref{lem_FHG:non_voters_prefer_pi*}  we have that 
\[\phi(\pi_{\XX},\pi)=\phi_{N\bs V}(\pi_{\XX},\pi)+\phi_V(\pi_{\XX},\pi)\ge(m+1)-m\ge 1.\qedhere\]
\end{proof}

\begin{lemma}
\label{lem_FHG:gadgets_are_valid}
Let $\DD_i$ be an assignment or gate gadget. 
Then either $T_i\subseteq \MC$ and $F_i\cap \MC=\emptyset$, or $F_i\subseteq \MC$ and $T_i\cap \MC=\emptyset$.
\end{lemma}
\begin{proof}
By \Cref{lem_FHG:clashing_sets} $\MC$ cannot contain agents from both sets $T_i$ and $F_i$.
Since this is true for all gadgets in the game, we have that $|\MC\cap ((G\cup A)\bs R)|\le \frac{1}{2}|(G\cup A)\bs R|$.
If additionally we have $(T_i\cup F_i)\cap \MC=\emptyset$, then $|\MC\cap ((G\cup A)\bs R)|<\frac{1}{2}|(G\cup A)\bs R|$, which implies that all agents in $Q$ prefer $\pi_{\XX}$ over $\pi$.
Hence, by \Cref{lem_FHG:non_voters_prefer_pi*} we have $\phi(\pi_{\XX},\pi)=\phi_{N\bs V}(\pi_{\XX},\pi)+\phi_V(\pi_{\XX},\pi)\ge(m+1)-m\ge 1$, a contradiction.
Therefore, $\MC$ must contain agents from exactly one of the sets $T_i$ and $F_i$. 
The claim then follows by \Cref{lem_FHG:replicas_with_origin}.
\end{proof}

\begin{lemma}
\label{lem_FHG:gadgets_comply}
All gate gadgets $\pi$-comply with their inputs (who are $\pi$-valid by \Cref{lem_FHG:gadgets_are_valid}).
\end{lemma}
\begin{proof}
Assume for contradiction that there exists a gate gadget $\GDD_i$ which does not $\pi$-comply with its input (or inputs). 
We consider two cases. 
For the first case, assume that $\GDD_i$ is an And-gadget with input gadgets $\GDD_j$ and $\GDD_k$, and that we have $\Fg_i\cup\Tg_j\cup\Tg_k\subseteq\MC$.
If all agents of $\Fg_i$ prefer $\pi_{\XX}$ over $\pi$, then by \Cref{lem_FHG:non_voters_prefer_pi*} we have $\phi(\pi_{\XX},\pi)=\phi_{N\bs V}(\pi_{\XX},\pi)+\phi_V(\pi_{\XX},\pi)\ge m+1-m\ge 1$ and we are done.
Otherwise, $\MC$ must contain at least two of $\{\wand_i,\zand_i,\yand_i\}$ (along with all respective replicas, by \Cref{lem_FHG:replicas_with_origin}).
Hence, there exist clashing sets $S_1,S_2\in\{\Wand_i,\Zand_i,\Yand_i\}$ such that $S_1\cup S_2\subseteq \MC$.
Thus, by \Cref{lem_FHG:R_bad_coalition} all replica agents in $S_1\cup S_2$ (of which there are at least $2m+2$) prefer $\pi_{\XX}$ over $\pi$.
Thus, by \Cref{lem_FHG:non_voters_prefer_pi*} we have
\[\phi(\pi_{\XX},\pi)=\phi_{N\bs V}(\pi_{\XX},\pi)+\phi_V(\pi_{\XX},\pi)\ge(2m+2)-m\ge 1.\]

In any other case, by \Cref{lem_FHG:gadgets_are_valid} and by definition of $\pi$-compliance, there exist a gate gadget $\GDD_i$ with an input gadget (gate or assignment) $\DD_j$, with clashing sets $S_i\in\{\Tg_i,\Fg_i\}$ and $S_j\in \{T_j,F_j\}$ such that $S_i\cup S_j\subseteq \MC$.
Thus, by \Cref{lem_FHG:R_bad_coalition} all replica agents in $S_i\cup S_j$ (of which there are at least $2m+2$) prefer $\pi_{\XX}$ over $\pi$.
Thus, by \Cref{lem_FHG:non_voters_prefer_pi*} we have
\[\phi(\pi_{\XX},\pi)=\phi_{N\bs V}(\pi_{\XX},\pi)+\phi_V(\pi_{\XX},\pi)\ge(2m+2)-m\ge 1.\qedhere\]
\end{proof}

\textbf{Extracting an assignment from $\pi$:}
Building upon \Cref{lem_FHG:gadgets_are_valid}, for the remainder of the section let $\XX'$ denote the assignment where $\XX'_i=1$ if and only if gadget $\ADD_i$ $\pi$-corresponds to $1$.

We proceed with establishing the coalitions of all remaining agents of the game, namely those of $L\cup W\cup Z\cup Y$, and also of $(G\cup A)\bs \MC$.
We will make use of the partition $\pi_{\XX'}$, the natural partition corresponding to $\XX'$ defined above.
We will first show that every agent is indifferent between $\pi$ and $\pi_{\XX'}$.
This will help to later establish that in fact $\pi=\pi_{\XX'}$.

\begin{lemma}
\label{lem_FHG:indifference_of_pi'_pi}
All agents are indifferent between $\pi$ and $\pi_{\XX'}$.
\end{lemma}
\begin{proof}
We show that no agent may prefer $\pi$ over $\pi_{\XX'}$; this suffices, since once this is established any agent who prefers $\pi_{\XX'}$ over $\pi$ would imply that $\pi_{\XX'}$ is a Pareto improvement from $\pi$, a contradiction to its Pareto optimality.

Notice that \Cref{lem_FHG:natural_utility_for_non_voters_1} holds for $\pi_{\XX'}$ as well, and so all non-voter agents obtain utility $1$ in $\pi_{\XX'}$.
Furthermore, by \Cref{lem_FHG:non_voters_prefer_pi*} they obtain at most $1$ in $\pi$.
Hence, agents from $N\bs V$ cannot prefer $\pi$ over $\pi_{\XX'}$.

Therefore, it suffices to show that voters cannot prefer $\pi$ over $\pi_{\XX'}$.
Notice that the deviation from $\pi$ to $\pi_{\XX'}$ neither adds nor removes any agents of $G\cup A$ into or out of $\MC$, since $\pi_{\XX'}$ is defined with respect to $\XX'$, and by \Cref{lem_FHG:replicas_with_origin,lem_FHG:gadgets_are_valid,lem_FHG:gadgets_comply} each gadget in the game also $\pi$-corresponds to the value that the corresponding gate is evaluated to by $\XX'$, while replicas are always with their origin; thus, by deviating from $\pi$ to $\pi_{\XX'}$, the only possible changes to $\MC$ are addition or removal of $W\cup Z\cup Y$-agents into or out of $\MC$.
Furthermore, notice that any voter agent $v$ assigns $\mv{v}$ to all $W\cup Z\cup Y$-agents.
Hence, $v$ may only prefer $\pi$ over $\pi_{\XX'}$ if the deviation extracts more agents from $\MC$ than it inserts (the reasoning for this is similar to the proof of \Cref{prop_FHG:maximal_util_welcome}, while also noting that $v$ is indifferent to swapping agents to whom she assigns the same value).
Let $\GDD_i$ be an And-gadget of a gate $\GG_i=\GG_j\land\GG_k$, and denote $t'=|\pi_{\XX'}(x_1)\cap (\Wand_i\cup\Zand_i\cup\Yand_i)|$ and $t=|\pi(x_1)\cap (\Wand_i\cup\Zand_i\cup\Yand_i)|$.
Notice that if $\XX'(\GG_j)=\XX'(\GG_k)=1$ then $t'=0$, and in any other case $t'=m+2$.
In each case, we obtain that $t$ cannot be greater than $t'$: if $\XX'(\GG_j)=\XX'(\GG_k)=1$ then any one of the sets $\Wand_i$, $\Zand_i$ or $\Yand_i$ clashes with at least one of $T_j$ or $T_k$, so by \Cref{lem_FHG:clashing_sets} we have $t=0$; otherwise, we clearly cannot have $t>m+2$, as any pair of sets from $\Wand_i$, $\Zand_i$ and $\Yand_i$ are clashing, again a contradiction to \Cref{lem_FHG:clashing_sets}.
Hence, no voter can prefer $\pi$ over $\pi_{\XX'}$, and we are done.
\end{proof}

\begin{lemma}
\label{lem_FHG:L_clean}
Let $\ell\in L$. 
Then $\countv{\ell}{0}{\pi}=0$.
\end{lemma}
\begin{proof}
First, assume that $\ell$ is not a replica agent.
Then by \Cref{lem_FHG:L_notin_MC,lem_FHG:gadgets_are_valid}, $\pi(\ell)$ contains at most one agent to whom she assigns $2$, and apart from that only agents to whom she assigns $0$ or $1$.
Thus, we have 
\[u_{\ell}(\pi)=
\frac{2\countv{\ell}{2}{\pi}+\countv{\ell}{1}{\pi}}{1+\countv{\ell}{2}{\pi}+\countv{\ell}{1}{\pi}+\countv{\ell}{0}{\pi}}
\le 
\frac{2\countv{\ell}{2}{\pi}+\countv{\ell}{1}{\pi}}{2+\countv{\ell}{2}{\pi}+\countv{\ell}{1}{\pi}}.
\]
It can be seen that for $\countv{\ell}{2}{\pi}\in\{0,1\}$, this quantity is less than $1$, a contradiction to \Cref{lem_FHG:natural_utility_for_non_voters_1,lem_FHG:indifference_of_pi'_pi}.

Since the claim holds for all non-replicas in $L$, by \Cref{lem_FHG:replicas_with_origin} it must also hold for the replicas since replicas only assign $0$ to an agent if their origin does so.
\end{proof}

\begin{lemma}
\label{lem_FHG:L_arrangement}
Let $\DD_i$ be an assignment or gate gadget, and let $L_i=(\Lg\cup\La)\cap\DD_i$ with $\ell_i$ denoting the non-replica agent of $L_i$.
If $F_i\subseteq \MC$ then $T_i\cup L_i\in \pi$, and if $T_i\subseteq \MC$ then $F_i\cup L_i\in \pi$.
\end{lemma}
\begin{proof}
First, note that by \Cref{lem_FHG:gadgets_are_valid} one of the above two cases must hold.
Assume that $F_i\subseteq \MC$. 
By \Cref{lem_FHG:L_notin_MC} we have that $L\cap\MC=\emptyset$.
Thus, if $T_i\cap \pi(\ell_i)=\emptyset$ then $\ell_i$ obtains utility less than $1$ in $\pi$, a contradiction to \Cref{lem_FHG:natural_utility_for_non_voters_1,lem_FHG:indifference_of_pi'_pi}.
Therefore, $T_i\cap \pi(\ell_i)\ne\emptyset$ and by \Cref{lem_FHG:replicas_with_origin} we have that $T_i\cup L_i\subseteq \pi(\ell_i)$.
Lastly, since $\ell_i$ assigns $0$ to any agent outside $L_i\cup T_i\cup F_i$, if $\pi(\ell_i)$ contains any additional agent she would prefer $\pi_{\XX'}$ over $\pi$, reaching the same contradiction as before.

The proof for the case $T_i\subseteq \MC$ is analogous.
\end{proof}

\begin{lemma}
\label{lem_FHG:WZY_arrangement}
Let $\GDD_i$ be an And-gadget of gate $\GG_i=\GG_j\land\GG_k$. 
Then the following statements hold.
\begin{enumerate}
    \item If $\{\gt_j,\gt_k,\gt_i\}\subseteq \MC$ then $\{\Wand_i\cup\Lw_i, \Zand_i\cup \Lz_i, \Yand_i\cup \Ly_i\}\subseteq \pi$.
    
    \item If $\{\gf_j,\gf_k,\gf_i\}\subseteq \MC$ then $\Yand_i\subseteq \MC$ and $\{\Wand_i\cup\Lw_i, \Zand_i\cup \Lz_i, \Ly_i\}\subseteq \pi$.

    \item If $\{\gf_j,\gt_k,\gf_i\}\subseteq \MC$ then $\Zand_i\subseteq \MC$ and $\{\Wand_i\cup\Lw_i, \Lz_i, \Yand_i\cup\Ly_i\}\subseteq \pi$.

    \item If $\{\gt_j,\gf_k,\gf_i\}\subseteq \MC$ then $\Wand_i\subseteq \MC$ and $\{\Lw_i, \Zand_i\cup \Lz_i, \Yand_i\cup\Ly_i\}\subseteq \pi$.
\end{enumerate}
\end{lemma}
\begin{proof}
First, note that by \Cref{lem_FHG:gadgets_are_valid,lem_FHG:gadgets_comply} the above enumeration lists all possible cases.
Moreover, by \Cref{lem_FHG:gadgets_are_valid,lem_FHG:L_clean,lem_FHG:L_arrangement}, for each $\ga\in \Wand_i\cup\Lw_i\cup \Zand_i\cup \Lz_i\cup \Yand_i\cup \Ly_i$, if $\ga\notin \MC$ then $\{\gt_i,\gf_i\}\cap\pi(\ga)=\emptyset$.
Call this observation $\spadesuit$.

If $\{\gt_j,\gt_k,\gt_i\}\subseteq \MC$ then by \Cref{lem_FHG:clashing_sets} we have $\MC\cap (\Wand_i\cup\Lw_i\cup \Zand_i\cup \Lz_i\cup \Yand_i\cup \Ly_i)=\emptyset$.
By this and by $\spadesuit$ we have that $\pi(\wand_i)$, $\pi(\zand_i)$, and $\pi(\yand_i)$ must contain $\lw_i$, $\lz_i$, and $\ly_i$ respectively, as otherwise one of those agents must prefer $\pi_{\XX'}$ over $\pi$, a contradiction to \Cref{lem_FHG:indifference_of_pi'_pi}.
\Cref{lem_FHG:replicas_with_origin,lem_FHG:L_clean} then conclude the first case.

For the second item, suppose that $\{\gf_j,\gf_k,\gf_i\}\subseteq \MC$.
Thus, by \Cref{lem_FHG:clashing_sets} we have $\MC\cap (\Wand_i\cup\Lw_i\cup \Zand_i\cup \Lz_i\cup \Ly_i)=\emptyset$.
By this and by $\spadesuit$ we have that $\pi(\wand_i)$ and $\pi(\zand_i)$ must contain $\lw_i$ and $\lz_i$ respectively, as otherwise one of those agents must prefer $\pi_{\XX'}$ over $\pi$, a contradiction to \Cref{lem_FHG:indifference_of_pi'_pi}.
\Cref{lem_FHG:replicas_with_origin,lem_FHG:L_clean} then imply that $\{\Wand_i\cup\Lw_i, \Zand_i\cup \Lz_i, \Ly_i\}\subseteq \pi$.
Furthermore, if $\yand_i\notin \MC$ then consider $\gf_i$.
She assigns $1.25$ to each of $\gf_j$ and $\gf_k$, and at most $1$ to any other agent in $\MC$.
Thus, her utility is less than $1$, a contradiction to \Cref{lem_FHG:natural_utility_for_non_voters_1,lem_FHG:indifference_of_pi'_pi}.
Therefore, $\yand_i\in \MC$, and by \Cref{lem_FHG:replicas_with_origin} $\Yand_i\subseteq \MC$.

For the third item, suppose that $\{\gf_j,\gt_k,\gf_i\}\subseteq \MC$.
Thus, by \Cref{lem_FHG:clashing_sets} we have $\MC\cap (\Wand_i\cup\Lw_i\cup \Lz_i\cup \Yand_i\cup\Ly_i)=\emptyset$.
By this and by $\spadesuit$ we have that $\pi(\wand_i)$ and $\pi(\yand_i)$ must contain $\lw_i$ and $\ly_i$ respectively, as otherwise one of those agents must prefer $\pi_{\XX'}$ over $\pi$, a contradiction to \Cref{lem_FHG:indifference_of_pi'_pi}.
\Cref{lem_FHG:replicas_with_origin,lem_FHG:L_clean} then imply that $\{\Wand_i\cup\Lw_i,\Lz_i, \Yand_i\cup\Ly_i\}\subseteq \pi$.
Furthermore, if $\zand_i\notin \MC$ then consider $\gf_i$.
She assigns $1.25$ to $\gf_j$, and at most $1$ to any other agent in $\MC$.
Thus, her utility is less than $1$, a contradiction to \Cref{lem_FHG:natural_utility_for_non_voters_1,lem_FHG:indifference_of_pi'_pi}.
Therefore, $\zand_i\in \MC$, and by \Cref{lem_FHG:replicas_with_origin} $\Zand_i\subseteq \MC$.

The fourth case where $\{\gt_j,\gf_k,\gf_i\}\subseteq \MC$ is analogous to the previous one.
\end{proof}

At this point, we can conclude that \Cref{lem_FHG:natural_partitions_beat_others} holds, as it follows directly from \Cref{lem_FHG:MC_X_V_Q,lem_FHG:gadgets_are_valid,lem_FHG:gadgets_comply,lem_FHG:L_arrangement,lem_FHG:WZY_arrangement}.
We can now directly prove that $\fCC$ admits a Condorcet string if and only if $\langle N,\mathbf{v}\rangle$ admits a strongly popular partition.

\textbf{Condorcet String implies Strongly Popular Partition:}
Suppose that we have a Condorcet string $\XX^*=(\XX^*_1,\dots,\XX^*_n)$ for $\fCC=\langle\CC_1,\dots,\CC_m\rangle$. 
We want to show that the corresponding natural partition $\pi^*:=\pi_{\XX^*}$ is strongly popular. 
To that end, assume towards contradiction there exists a partition $\pi\neq\pi^*$ such that $\phi(\pi^*,\pi)\leq 0$.
We may assume without loss of generality that $\pi$ is Pareto-optimal (see \Cref{footnote:Pareto}).
By \Cref{lem_FHG:natural_partitions_beat_others}, we have that $\pi=\pi_{\XX'}$ for some string $\XX'\in \bits{n}$.
Thus, we have that $\XX'\ne \XX^*$, as otherwise we would have $\pi^*=\pi$, a contradiction to the choice of $\pi$.
Furthermore, all non-voter agents are indifferent between $\pi^*$ and $\pi_{\XX'}$, as they obtain utility $1$ in both by \Cref{lem_FHG:natural_utility_for_non_voters_1}.

Consider a voter agent $v$ corresponding to a circuit $\CC$ with outputs $\OO_1,\dots,\OO_n$.
By \Cref{def_FHG:natural_partition}, both $\MC_{\XX'}$ and $\MC_{\XX^*}$ contain $X$, $V$, $Q$.
By \Cref{lem_FHG:gadgets_are_valid}, they additionally both contain exactly $2m+4$ agents from each And-gadget, and $m+2$ agents from each other gadget.
Thus, we have that $|\MC_{\XX'}|=|\MC_{\XX^*}|$.
Furthermore, $v$ assigns $2^n$ to all agents in $\MC_{\XX'}\bs \{\ot_j,\of_j\}_{j\in[n]}$.
Denote $O=\MC_{\XX'}\cap\{\ot_j,\of_j\}_{j\in[n]}$ and $O^*=\MC_{\XX^*}\cap\{\ot_j,\of_j\}_{j\in[n]}$.
We thus have
\[
u_{v}(\pi)=\frac{2^n\cdot (|\MC_{\XX'}|-n-1)+\sum_{\ot_j\in O}2^j}{|\MC_{\XX'}|}, \text{ and }
u_{v}(\pi^*)=\frac{2^n\cdot (|\MC_{\XX'}|-n-1)+\sum_{\ot_j\in O^*}2^j}{|\MC_{\XX'}|}.
\]
Hence, we have that the preference order of $v$ depends exclusively on $\sgn(\sum_{\ot_j\in O^*}2^j-\sum_{\ot_j\in O}2^j)$.
By definition of a natural partition, for all $j\in [n]$ we have that $\ot_j\in \MC_{\XX'}$ (and thus $\ot_j\in O$) if and only if $\XX'(\OO_j)=1$, and similarly $\ot_j\in \MC_{\XX^*}$ (and thus $\ot_j\in O^*$) if and only if $\XX^*(\OO_j)=1$.
Thus, we have that $\sum_{\ot_j\in O^*}2^j>\sum_{\ot_j\in O}2^j$ if and only if $\XX^*(\CC_j)>\XX'(\CC_j)$ (this follows from the nature of the powers of two, each of which exceeds the sum of all smaller powers, mimicking the behavior of more significant bits of a bit string).
Applying this logic across all circuits, since $\XX^*$ is a Condorcet string there are more voters who prefer $\pi^*$ over $\pi$ than those who prefer the opposite. 
Hence, we have that $\phi(\pi^*,\pi)>0$, a contradiction.
Thus, $\pi^*$ is a strongly popular partition.
\qed

\textbf{Strongly Popular Partition implies Condorcet String:}
Suppose now that we have a strongly popular partition $\pi^*$ for $\langle N,\mathbf{v}\rangle$.
We wish to show that there exists a Condorcet string for $\fCC$.
A strongly popular partition is also Pareto optimal, and therefore by \Cref{lem_FHG:natural_partitions_beat_others} there exists a bit string $\XX^*\in\bits{n}$ such that $\pi^*=\pi_{\XX^*}$ is the natural partition corresponding to $\XX^*$.
We wish to show that $\XX^*$ is a Condorcet string.
Assume towards contradiction that there exists a string $\XX^*\ne\XX'\in\bits{n}$ such that $\XX'$ beats or ties with $\XX^*$ in a majority vote among the circuits $\CC_1,\dots,\CC_m$.
Then by reasoning analogous to the forward direction, there are at least as many voters who prefer $\pi_{\XX'}$ over $\pi^*$ than those who prefer the opposite, while all other agents are indifferent between the two partitions. 
Hence, $\pi^*$ is not more popular than $\pi_{\XX'}$, a contradiction to the choice of $\pi^*$.
Therefore, $\XX^*$ is a Condorcet string.
\qed


\section{Strong Popularity in Non-Negative ASHGs}
\label{sec:ASHG_NN}
In this section, we prove \Cref{thm:ASHG_nonnegative}.

\begin{proof}
Let $\mathcal{H}=\langle N,\mathbf{v}\rangle$ be a non-negative ASHG.
Consider the undirected graph $G=(N,E)$ with agent set $N$ and $E=\{(i,j)\colon\mathbf{v}_i(j)>0 \lor \mathbf{v}_j(i)>0\}$.

\begin{observation}
\label{obs_ASHG_NN:edges}
Since $\mathbf{v}$ is non-negative, it holds that $(i,j)\notin E$ if and only if $\mathbf{v}_i(j)=\mathbf{v}_j(i)=0$.
\end{observation}

\begin{lemma}
\label{lem_ASHG_NN:connectivity}
There exists a strongly popular partition if and only if $G$ is connected.    
\end{lemma}

Establishing \Cref{lem_ASHG_NN:connectivity} completes the proof of \Cref{thm:ASHG_nonnegative}, since graph connectivity can be tested in polynomial time.

It remains to prove the lemma.
Let $\pi^*:=\{N\}$ denote the partition consisting solely of the grand coalition.
Notice that, given some partition $\pi$, every agent $p$ satisfies $\phi_p(\pi^*,\pi)\ge 0$, as she cannot obtain a utility greater than what she obtains in $\pi^*$.
Thus, no partition $\pi$ is more popular than $\pi^*$, implying that $\pi^*$ is the only candidate for a strongly popular partition.
It follows that to preclude the existence of a strongly popular partition it suffices to exhibit a partition that ties with $\pi^*$.

Now, suppose that $\pi^*$ is strongly popular.
Assume towards contradiction that $G$ is not connected, and let $S$ be a connected component of $G$.
Then by \Cref{obs_ASHG_NN:edges}, for all $i\in S$ and $j\in N\bs S$ we have that $\mathbf{v}_i(j)=\mathbf{v}_j(i)=0$.
Hence, the partition $\pi'=\{S,N\bs S\}$ is as popular as $\pi^*$, since all agents obtain the same utility in both partitions.
Indeed, $\pi^*$ is obtained from $\pi'$ by merging $S$ and $N\bs S$, but since all agent-valuations crossing the cut have weight $0$, the utility of an agent cannot increase from this merge. 
Furthermore, we have that $\pi'\ne\pi^*$ because $G$ is not connected, implying that $N\bs S\ne\emptyset$.
Thus, $\pi^*$ is not strongly popular, implying no strongly popular partition exists, a contradiction.
Hence, $G$ must be connected.

Conversely, suppose that $\pi^*$ is not strongly popular, namely there exists a partition $\pi'\ne\pi^*$ such that $\phi(\pi^*,\pi')\le 0$.
Then, it must be that every agent obtains the same utility in both partitions (since $\pi^*$ already yields maximal utility for all agents).
Therefore, all agents obtain maximal utility in $\pi'$.
Hence, for any two agents $i,j\in N$ such that $\pi'(i)\ne\pi'(j)$, we have that $\mathbf{v}_i(j)=\mathbf{v}_j(i)=0$, as otherwise merging the coalitions $\pi'(i)$ and $\pi'(j)$ yields better utility for at least one of them, a contradiction to the above.
Thus, by \Cref{obs_ASHG_NN:edges}, we have that, in $G$, there are no edges crossing between two different coalitions of $\pi'$.
Furthermore, since $\pi'\ne\pi^*$, we have that $\pi'$ contains at least two coalitions, implying that $G$ is not connected, concluding the proof.
\end{proof}

\section{Discussion}
We have shown that deciding whether an additively separable or fractional hedonic game admits a strongly popular partition is \PCW{}-complete.
These are the first natural complete problems for \PCW{}, and also, to our knowledge, the first completeness result for strong popularity in hedonic games.
Our results highlight an intriguing difference between strong popularity and other solution concepts such as Nash, individual, and contractual stability, as well as weak popularity and the core, for which the corresponding existence question is typically complete for \NP{} or \Sig{2} in various cardinal hedonic games.
Here, we illustrate that unambiguity calls for a different complexity-theoretic approach.
We conjecture that a similar characterization could be obtained for adjacent models, such as modified-fractional hedonic games.

\section*{Acknowledgements}
A preliminary version of \Cref{thm:ASHG_general} appeared as part of earlier versions of \cite{GGKN2025unambiguity}, when the two works formed a single manuscript.
We thank Paul W. Goldberg, Elias Koutsoupias, and Martin Bullinger for many insightful discussions and comments.
We thank Gencer Mert for spotting a technical error in a previous version of the proof of \Cref{thm:ASHG_general}.
The author was supported by the Engineering and Physical Sciences Research Council (EPSRC, grant EP/W524311/1).

\bibliographystyle{alpha}
\bibliography{Strong_Popularity_bibliography}

\end{document}